\documentclass[11pt]{article}
\usepackage[margin=1in]{geometry}
\usepackage{graphicx,tikz, color}

\usepackage{soul}
\usepackage{amsmath,amssymb,amsthm}
\usepackage{thm-restate}

\usepackage[table]{xcolor}
\definecolor{linkblue}{RGB}{0,70,140}
\definecolor{citegreen}{RGB}{0,110,90}
\definecolor{urlblue}{RGB}{0,90,160}

\usepackage[colorlinks=true,
            linkcolor=linkblue,
            citecolor=citegreen,
            urlcolor=urlblue]{hyperref}
            
\usepackage[capitalise,nameinlink]{cleveref}

\usepackage{enumerate}
\usepackage{framed}
\usepackage{ifthen}
\usepackage{pgf}

\usepackage{booktabs}
 
\usepackage{tabularx}

\usepackage[colorinlistoftodos,prependcaption,textsize=tiny]{todonotes}

\usepackage{pgfpages}

\usepackage{tcolorbox}
\tcbuselibrary{theorems}
\usepackage{thmtools}
\usepackage{thm-restate}

\usepackage{algorithm}
\usepackage{algpseudocode}
\usepackage{enumitem}
\usepackage[square,numbers]{natbib}
\usepackage{appendix}
\usepackage{soul}

\newtheorem{theorem}{Theorem}
\newtheorem*{theorem*}{Theorem}

\newtheorem{corollary}{Corollary}

\newtheorem{definition}{Definition}

\newtheorem{lemma}{Lemma}

\newcommand{\LOCAL}{\mathsf{LOCAL}}

\DeclareMathOperator{\bad}{bad}

\DeclareMathOperator{\poly}{poly}
\DeclareMathOperator{\polylog}{polylog}

\DeclareMathOperator{\dist}{dist}

\definecolor{shadecolor}{gray}{0.875}

\crefname{appendix}{appendix}{appendices}
\Crefname{appendix}{Appendix}{Appendices}

\allowdisplaybreaks[4]

\title{Introvert Clustering for Distributed Graph Algorithms}
\author{Yi-Jun Chang\footnote{National University of Singapore. ORCID: 0000-0002-0109-2432. Email: cyijun@nus.edu.sg}  \and Nima Dolatabadi\footnote{University of Copenhagen. Supported by VILLUM Foundation grant 54451, Basic Algorithms Research Copenhagen (BARC). Most of this work was completed while the author was at the National University of Singapore. ORCID: 0009-0000-0928-7499. Email: nima.dolatabadi@di.ku.dk}}
\date{}

\begin{document}

\maketitle

\begin{abstract}
We introduce a new graph decomposition primitive that we call
\emph{introvert clustering}.
It strengthens standard low-diameter clustering with an additional local
guarantee: every clustered vertex keeps at least a
$\left(\frac12-\varepsilon\right)$-fraction of its relevant neighbors inside its
own cluster---hence the name \emph{introvert}.
Repeatedly applying this primitive to the remaining edges yields a layered
introvert network decomposition with $O(\log n)$ layers and weak diameter
$O(\log n)$.

We give two applications of this decomposition in the $\LOCAL$ model.
First, for every constant $\varepsilon>0$, we obtain a
$\widetilde O(\log^2 n)$-round deterministic algorithm for list
$\left(\frac32+\varepsilon\right)\Delta$-edge coloring on graphs of maximum degree
$\Delta\geq\Delta_0(\varepsilon)$.
For bipartite graphs, the result holds for all $\Delta$.
Second, for every constant $0<\varepsilon<1/4$, we obtain a
$\widetilde O(\log^2 n)$-round deterministic algorithm for computing a
$\left(\frac14-\varepsilon\right)$-locally balanced cut, in which every vertex has
at least a $\left(\frac14-\varepsilon\right)$-fraction of its neighbors on the
opposite side.

The algorithms resulting from the decomposition are remarkably simple.
For edge coloring, we process the layers in reverse order and color each
cluster; for locally balanced cut, we process them in forward order and
compute a locally maximum cut inside each cluster.
The introvert guarantee is what makes these simple procedures work beyond
the usual greedy regime of network decomposition.

We construct the layered decomposition in $O(\log^2 n)$ randomized
rounds by combining the Miller--Peng--Xu low-diameter clustering with a
simple trimming procedure.
We also give a deterministic $\widetilde O(\log^2 n)$-round construction
through a white-box adaptation of the recursive network decomposition
algorithm of Ghaffari and Grunau [FOCS 2024].
\end{abstract}
\thispagestyle{empty}
\newpage
\bigskip
\tableofcontents
\bigskip
\thispagestyle{empty}

\newpage
\pagenumbering{arabic}

\section{Introduction}

We work in the standard $\LOCAL$ model of distributed computing~\cite{linial}.
The communication network is a simple, undirected, $n$-vertex connected graph $G=(V,E)$.
Each vertex hosts a processor with a unique identifier of
$O(\log n)$ bits.
Computation proceeds in synchronous rounds.
In each round, every vertex may perform arbitrary local computation and
exchange messages of unbounded size with each of its neighbors.
Initially, each vertex knows its own identifier and its incident edges,
together with the input associated with them; for example, in the list
edge coloring problem, the endpoints of an edge know its list of available
colors.
The round complexity of an algorithm is the number of communication
rounds until all vertices have produced their outputs.

We consider both deterministic and randomized algorithms.
In a randomized algorithm, vertices have access to private random bits,
and unless stated otherwise, the algorithm is required to succeed with
high probability, meaning with probability at least $1-1/\poly(n)$.
Throughout the paper, $n$ denotes the number of vertices of the underlying
graph, and $\widetilde O(\cdot)$ suppresses $\poly(\log\log n)$ factors.

\subsection{List Edge Coloring}
Edge coloring is one of the classic local symmetry-breaking problems in
distributed computing, alongside maximal independent set, maximal
matching, and vertex coloring.
Given a graph $G=(V,E)$ of maximum degree $\Delta$, the goal is to assign
colors to the edges so that any two edges sharing an endpoint receive
different colors.
Vizing's celebrated theorem~\cite{Vizing64} guarantees that every simple
graph admits a proper edge coloring with at most $\Delta+1$ colors.
In the distributed setting, however, the complexity of edge coloring
depends crucially on how many colors are available.

We study the more general list version of the problem.
In \emph{list $k$-edge coloring}, every edge $e$ is given a list $L(e)$
of at least $k$ available colors, and the goal is to compute a proper edge
coloring in which each edge receives a color from its own list.
Ordinary $k$-edge coloring is the special case in which all edges have
the same list of $k$ colors.

\paragraph{The greedy threshold.}
The palette size $2\Delta-1$ forms a natural threshold for distributed edge
coloring.
Indeed, in any partial edge coloring, an uncolored edge is adjacent to at
most $2\Delta-2$ already colored edges and hence can always be colored
greedily if $2\Delta-1$ colors are available.
This regime has been studied extensively
\cite{PanconesiS97,PanconesiRizzi01,EPS15,GhaffariSu17,
GhaffariHKMSU17,FGK17,GHK18,Harris19,Kuhn20,BalliuKO20,BBKO22}.
A major breakthrough was the work of Fischer, Ghaffari, and
Kuhn~\cite{FGK17}, which gave the first deterministic
$\poly(\log n)$-round algorithm using exactly $2\Delta-1$ colors;
moreover, their algorithm already works for list edge coloring.
A recent line of work progressively improved the dependence on
$\Delta$~\cite{Kuhn20,BalliuKO20,BBKO22}, culminating in the
$\poly(\log\Delta)+O(\log^*n)$-round deterministic algorithm of
Balliu, Brandt, Kuhn, and Olivetti~\cite{BBKO22}.

\paragraph{Below the greedy threshold.}
The picture changes fundamentally below $2\Delta-1$.
Chang, He, Li, Pettie, and Uitto~\cite{journals/talg/ChangHLPU20}
showed that even $(2\Delta-2)$-edge coloring requires
$\Omega(\log_\Delta n)$ deterministic rounds, even on trees.

The first deterministic polylogarithmic-round algorithms substantially
below this threshold were given by Ghaffari, Kuhn, Maus, and
Uitto~\cite{GKMU18}.
They obtained a $3\Delta/2$-edge coloring in
$\poly(\Delta,\log n)$ rounds.
The $3\Delta/2$ regime has since received considerable attention.
Brandt, Maus, Narayanan, Schager, and Uitto~\cite{BrandtMNSU25}
substantially improved the complexity using new local algorithms based on
Hall's theorem, obtaining a $3\Delta/2$-edge coloring in
$O(\Delta^2\log n)$ rounds and a
$(\frac32+\varepsilon)\Delta$-edge coloring in
$\widetilde O(\varepsilon^{-2}\log^2\Delta\log n)$ rounds.
More recently, Maus, Nolin, and Schager~\cite{MausNS26} improved the latter
bound to
$O(\varepsilon^{-1}\log^2\Delta\log n+\varepsilon^{-2}\log n)$ rounds.

There has also been substantial progress with palettes smaller than
$3\Delta/2$.
Ghaffari, Kuhn, Maus, and Uitto~\cite{GKMU18} gave an
$(1+\varepsilon)\Delta$-edge coloring in
$\poly(\varepsilon^{-1},\log n)$ rounds when
$\Delta=\Omega(\varepsilon^{-1}\log(1/\varepsilon)\log n)$.
A major step toward Vizing's bound was made by Su and Vu~\cite{SuVu19},
who gave a randomized $(\Delta+2)$-edge coloring in
$\poly(\Delta,\log n)$ rounds.
Bernshteyn~\cite{Bernshteyn22} subsequently obtained the first
deterministic $\poly(\Delta,\log n)$-round algorithm using the optimal
$\Delta+1$ colors.
The dependence on $n$ was later
improved~\cite{conf/stoc/Christiansen23,BernshteynDhawan25}.

Very recently, de Vos, Maus, and Blikstad~\cite{DeVosMB26} obtained
deterministic $\LOCAL$ algorithms for
$(1+\varepsilon)\Delta+O(\sqrt{\log n})$-edge coloring, including an
$\widetilde O(\log^2 n)$-round algorithm.

\paragraph{List edge coloring.}
The situation is different for the more general list edge coloring problem.
At the greedy threshold $2 \Delta - 1$, many deterministic algorithms already work for list edge coloring; see, e.g., \cite{FGK17,GHK18,Harris19,Kuhn20,BalliuKO20}.
Below the greedy threshold, however, much less is known.

On the existential side, Kahn~\cite{Kahn96} proved that, for every
$\varepsilon>0$, there exists a constant $\Delta_0(\varepsilon)$ such that
every graph of maximum degree $\Delta\geq\Delta_0(\varepsilon)$ admits a
list $(1+\varepsilon)\Delta$-edge coloring.
More recently, Bonamy, Delcourt, Lang, and Postle~\cite{BonamyDLP24}
proved a local strengthening: if the minimum degree is at least
$\ln^{25}\Delta$, then lists of size
$(1+\varepsilon)\max\{\deg(u),\deg(v)\}$ suffice for every edge
$\{u,v\}$.
This local form will be an important ingredient in our algorithm.

Randomized distributed list edge coloring below the greedy threshold is
also known.
For every fixed $\varepsilon>0$, Elkin, Pettie, and Su~\cite{EPS15}
gave a randomized list $(1+\varepsilon)\Delta$-edge coloring algorithm
for $\Delta\geq\Delta_0(\varepsilon)$ with round complexity
$O\left(\left\lceil\frac{\log n}{\Delta^{1-\gamma}}\right\rceil
+\log^*\Delta\right)$ for any fixed constant $\gamma>0$.
Chang, He, Li, Pettie, and Uitto~\cite{journals/talg/ChangHLPU20}
further obtained randomized $(1+\varepsilon)\Delta$-edge coloring
algorithms for non-constant $\varepsilon$, but for ordinary rather than
list edge coloring.

Our main result is a new deterministic list edge coloring algorithm below the
greedy threshold.

\begin{restatable}[List edge coloring]{theorem}{mainthm}
\label{thm:intro-main}
For every constant $\varepsilon>0$, there exists a constant
$\Delta_0=\Delta_0(\varepsilon)$ such that list
$(\frac32+\varepsilon)\Delta$-edge coloring on graphs of maximum degree
$\Delta\geq\Delta_0$ can be solved in $O(\log^2 n)$ randomized rounds
with high probability and in $\widetilde O(\log^2 n)$ deterministic
rounds in the $\LOCAL$ model.
For bipartite graphs, the result holds for all $\Delta$.
\end{restatable}

For comparison, combining the recent results of de Vos, Maus, and
Blikstad~\cite{DeVosMB26} and Maus, Nolin, and
Schager~\cite{MausNS26} gives a deterministic
$(\frac32+\varepsilon)\Delta$-edge coloring algorithm in
$\widetilde O(\log^2 n)$ rounds.
Their results, however, concern ordinary edge coloring, where all edges
share a common palette.
Our result achieves the same round complexity for the
strictly more general list edge coloring problem, where each edge has
its own list of available colors.

To the best of our knowledge, no previous work gives a dedicated
deterministic distributed algorithm for list edge coloring below the
greedy threshold.
One can obtain such an algorithm indirectly by derandomizing the randomized
algorithm of Elkin, Pettie, and Su~\cite{EPS15} using deterministic network
decompositions~\cite{ghaffari_grunau_focs24}.
For any fixed constant $\gamma>0$, this gives
\(
\widetilde O\left(
\left\lceil\frac{\log n}{\Delta^{1-\gamma}}\right\rceil
\log^2 n
\right)
\)
rounds.
While this is $\widetilde O(\log^2 n)$ when
$\Delta\geq(\log n)^{1/(1-\gamma)}$, it becomes
$\widetilde O(\log^3 n)$ for constant $\Delta$.

By allowing lists of size $(\frac32+\varepsilon)\Delta$, we obtain
$\widetilde O(\log^2 n)$ rounds uniformly throughout the entire range
$\Delta\geq\Delta_0(\varepsilon)$.
Thus, at this palette size, list edge coloring matches the
$\widetilde O(\log^2 n)$ deterministic complexity currently achievable
for ordinary edge coloring.
Refer to \Cref{tab:edge-coloring-comparison} for a comparison with prior
work.

\begin{table}[ht!]
\centering
\caption{Comparison with prior work on edge coloring.}
\label{tab:edge-coloring-comparison}
\small
\setlength{\tabcolsep}{5pt}
\renewcommand{\arraystretch}{1.6}
\begin{tabular}{llll}
\toprule
Colors & Type & Rounds & Reference \\
\midrule

\multicolumn{4}{@{}l}{\emph{Ordinary edge coloring}} \\[2pt]

$(1+\varepsilon)\Delta+O(\sqrt{\log n})$
& Deterministic
& $O(\log^2 n)+\widetilde O(\log^2\Delta\log n)$
& \cite{DeVosMB26}
\\

$(1+\varepsilon)\Delta+O(\sqrt{\log n})$
& Deterministic
& $\widetilde O(\log^2 n)$
& \cite{DeVosMB26}
\\

$3\Delta/2$
& Deterministic
& $O(\Delta^9\polylog n)$
& \cite{GKMU18}
\\

$3\Delta/2$
& Deterministic
& $\widetilde O(\Delta^4\log^6 n)$
& \cite{Harris19}
\\

$3\Delta/2$
& Deterministic
& $O(\Delta^2\log n)$
& \cite{BrandtMNSU25}
\\

$(\frac32+\varepsilon)\Delta$
& Deterministic
& $\widetilde O(\varepsilon^{-2}\log^2\Delta\log n)$
& \cite{BrandtMNSU25}
\\

$(\frac32+\varepsilon)\Delta$
& Deterministic
& $O(\varepsilon^{-1}\log^2\Delta\log n
+\varepsilon^{-2}\log n)$
& \cite{MausNS26}
\\

$(\frac32+\varepsilon)\Delta$
& Deterministic
& $\widetilde O(\log^2 n)$
& \cite{MausNS26} + \cite{DeVosMB26}
\\

\midrule

\multicolumn{4}{@{}l}{
\emph{List edge coloring}
($\varepsilon>0$ and $\gamma>0$ are any constants,
$\Delta\geq\Delta_0(\varepsilon)$)
} \\[2pt]

$(1+\varepsilon)\Delta$
& Randomized
& $O\left(
\left\lceil\frac{\log n}{\Delta^{1-\gamma}}\right\rceil
+\log^*\Delta
\right)$
& \cite{EPS15}
\\

$(1+\varepsilon)\Delta$
& Deterministic
& $\widetilde O\left(
\left\lceil\frac{\log n}{\Delta^{1-\gamma}}\right\rceil
\log^2 n
\right)$
& \cite{EPS15,ghaffari_grunau_focs24}
\\

\rowcolor{gray!15}
$(\frac32+\varepsilon)\Delta$
& Randomized
& $O(\log^2 n)$
& \Cref{thm:intro-main}
\\

\rowcolor{gray!15}
$(\frac32+\varepsilon)\Delta$
& Deterministic
& $\widetilde O(\log^2 n)$
& \Cref{thm:intro-main}
\\

\bottomrule
\end{tabular}
\end{table}

\subsection{Locally Balanced Cut}
Our techniques also give a new result for the locally balanced cut problem.
A cut is \emph{locally maximum} if moving any single vertex to the other
side cannot increase the number of crossing edges, or equivalently, if
every vertex has at least half of its neighbors on the opposite side.
The problem is naturally amenable to local search: starting from an
arbitrary cut, repeatedly moving any vertex that improves the cut
eventually reaches a locally maximum cut.
It has been studied extensively in local search and algorithmic game
theory~\cite{SchafferY91,AngelBPW17}.

Perhaps surprisingly for such a fundamental problem, its distributed
complexity was largely open until very recently.
Balliu, Boudier, d'Amore, Kuhn, Olivetti, Schmid, and
Suomela~\cite{Balliu26Potential} gave the first nontrivial distributed
upper bound: an $O(\Delta^2\log^6 n)$-round randomized algorithm and an
$O(\Delta^2) \cdot \widetilde{O}(\log^8 n)$-round deterministic algorithm via derandomization.
They also proved an $\Omega(\min\{\Delta,\sqrt n\})$ lower bound, even in
the quantum $\LOCAL$ model.
Thus, a polynomial dependence on $\Delta$ is inherent for finding an exact locally
maximum cut.

We show that this dependence can be avoided if we relax the local
optimality requirement.
For $0<\alpha\leq 1/2$, we call a cut \emph{$\alpha$-locally balanced}
if every vertex $v$ has at least $\alpha\deg(v)$ neighbors on the
opposite side.
Thus, a $1/2$-locally balanced cut is precisely a locally maximum cut.
Unlike a global approximation to maximum cut, this relaxation 
provides a guarantee at every individual vertex.

\begin{restatable}[Locally balanced cut]{theorem}{approxcutthm}
\label{thm:intro-cut}
For every constant $0<\varepsilon<1/4$, a
$\left(\frac14-\varepsilon\right)$-locally balanced cut can be computed in
$O(\log^2 n)$ randomized rounds with high probability and in
$\widetilde O(\log^2 n)$ deterministic rounds in the $\LOCAL$ model.
\end{restatable}

Thus, relaxing the local guarantee from $1/2$ to
$\frac14-\varepsilon$ changes the complexity substantially: the
polynomial dependence on $\Delta$ inherent for locally maximum cut
disappears, and the problem becomes solvable in polylogarithmic time
independently of $\Delta$.

At the same time, the relaxed problem remains nontrivial.
A closely related problem is \emph{$k$-partial $2$-coloring}, in which
the vertices are colored with two colors so that every vertex has at
least $k$ neighbors of the opposite color.
Balliu, Hirvonen, Lenzen, Olivetti, and Suomela~\cite{BalliuHLOS19}
showed that $2$-partial $2$-coloring requires $\Omega(\log n)$
deterministic rounds and $\Omega(\log\log n)$ randomized rounds on
$d$-regular trees, for every constant $d\geq2$.
These bounds imply the same lower bounds for every constant
$\alpha>0$: choosing a sufficiently large constant $d$ with
$\alpha d>1$, every $\alpha$-locally balanced cut on a $d$-regular
graph is also a $2$-partial $2$-coloring.

Consequently, for every constant $0<\varepsilon<1/4$, the deterministic
complexity of
$\left(\frac14-\varepsilon\right)$-locally balanced cut lies between
$\Omega(\log n)$ and $\widetilde O(\log^2 n)$.
See \Cref{tab:locally-balanced-cut} for a comparison with prior work.

\begin{table}[ht!]
\centering
\caption{Comparison of our results with prior work on locally balanced cuts.}
\label{tab:locally-balanced-cut}
\small
\setlength{\tabcolsep}{6pt}
\renewcommand{\arraystretch}{1.2}
\begin{tabular}{llll}
\toprule
Local guarantee & Type & Rounds & Reference \\
\midrule

Any constant $\alpha>0$
& Randomized lower bound
& $\Omega(\log\log n)$
& \cite{BalliuHLOS19}
\\

Any constant $\alpha>0$
& Deterministic lower bound
& $\Omega(\log n)$
& \cite{BalliuHLOS19}
\\

\midrule

$\alpha=1/2$
& Quantum lower bound
& $\Omega(\min\{\Delta,\sqrt n\})$
& \cite{Balliu26Potential}
\\

$\alpha=1/2$
& Randomized upper bound
& $O(\Delta^2\log^6 n)$
& \cite{Balliu26Potential}
\\

$\alpha=1/2$
& Deterministic upper bound
& $O(\Delta^2) \cdot \widetilde{O}(\log^8 n)$
& \cite{Balliu26Potential}
\\

\midrule

\rowcolor{gray!15}
$\alpha=1/4-\varepsilon$
& Randomized upper bound
& $O(\log^2 n)$
& \Cref{thm:intro-cut}
\\

\rowcolor{gray!15}
$\alpha=1/4-\varepsilon$
& Deterministic upper bound
& $\widetilde O(\log^2 n)$
& \Cref{thm:intro-cut}
\\

\bottomrule
\end{tabular}
\end{table}

\subsection{Our Technique: Introvert Clustering}

The common ingredient behind both of our results is a new clustering
primitive that we call \emph{introvert clustering}.
The main idea is to strengthen low-diameter clustering with a local
guarantee: every clustered vertex keeps almost half of its incident edges
inside its own cluster.
We first describe this primitive and its randomized and deterministic
constructions, and then explain its applications to list edge coloring and
locally balanced cuts.

\paragraph{Clustering terminology.}
Let $G=(V,E)$ be a graph.
For $S\subseteq V$, we write $G[S]$ for the subgraph of $G$ induced by
$S$, $N_G(v)$ for the set of neighbors of $v$ in $G$, and
$\deg_G(v)=|N_G(v)|$.

A \emph{clustering} of $G$ is a collection $\mathcal C$ of pairwise
disjoint nonempty subsets of $V$, which we call \emph{clusters}.
A vertex belonging to some cluster is \emph{clustered}, and all other
vertices are \emph{unclustered}.
An edge is \emph{intra-cluster} if its two endpoints belong to the same
cluster; all other edges are called \emph{inter-cluster}.
Thus, edges incident to unclustered vertices are also inter-cluster.

For a cluster $C\subseteq V$, its \emph{weak diameter} is
$\max_{u,v\in C}\dist_G(u,v)$, whereas its \emph{strong diameter} is
$\max_{u,v\in C}\dist_{G[C]}(u,v)$.
Thus, paths witnessing weak diameter may leave the cluster, while paths
witnessing strong diameter must stay inside the cluster.
Throughout this paper, we only require weak diameter.

\begin{definition}[Introvert clustering]
\label{def:introvert-clustering}
Let $G=(V,E)$ be a graph. Let $\varepsilon>0$, $\beta\in(0,1)$, and
 $D\geq 1$.
A clustering $\mathcal C$ of $G$ is an
\underline{$(\varepsilon,\beta,D)$-introvert clustering} if it satisfies the
following properties.
\begin{description}
    \item[Few inter-cluster edges.]
    At most $\beta|E|$ edges are inter-cluster.
    
    \item[Low diameter.]
    Every cluster $C\in\mathcal C$ has weak diameter at most $D$.

    \item[Introvert property.]
    For every $C\in\mathcal C$ and every $v\in C$,
    $|N_G(v)\cap C|\geq\left(\frac12-\varepsilon\right)\deg_G(v)$.
\end{description}
\end{definition}

If we drop the introvert property, the remaining two conditions are the
standard type of low-diameter clustering guarantee provided by the
Miller--Peng--Xu decomposition~\cite{MPX13SPAA}: the clusters have small
diameter and only a small fraction of the edges are inter-cluster.
The new requirement is local.
A standard low-diameter clustering may have few inter-cluster edges
overall while a particular vertex has almost all of its incident edges
leaving its cluster.
Introvert clustering rules out this behavior for every clustered vertex.

A single introvert clustering does not assign every edge to a cluster.
To cover all edges, we repeatedly apply introvert clustering to the edges
that remain after the previous layers.

\begin{definition}[Layered introvert network decomposition]
\label{def:layered-introvert}
An \ul{$(\varepsilon,\ell,D)$-layered introvert network decomposition} of a graph
$G=(V,E)$ consists of a partition of $E$ into $\ell$ disjoint edge sets $E = E_1 \cup \cdots \cup E_\ell$.
For each $i\in[\ell]$, define
$G_i=(V,E_i\cup E_{i+1}\cup\cdots\cup E_\ell)$.
The decomposition also specifies a clustering
$\mathcal C_i$ of $G_i$ for each $i\in[\ell]$ satisfying the following
properties.
\begin{description}
    \item[Layer assignment.]
    $E_i$ is exactly the set of intra-cluster edges of $\mathcal C_i$
    in $G_i$.

    \item[Low diameter.]
    Every cluster $C\in\mathcal C_i$ has weak diameter at most $D$
    in $G$.

    \item[Introvert property.]
    For every $C\in\mathcal C_i$ and every $v\in C$,
    $|N_{G_i}(v)\cap C|
    \geq\left(\frac12-\varepsilon\right)\deg_{G_i}(v)$.
\end{description}
\end{definition}

Intuitively, $G_i$ consists of the edges that remain at the beginning of
layer $i$.
Every edge is assigned to exactly one layer, whereas a vertex may belong
to clusters in several different layers.
The introvert property at layer $i$ is measured with respect to $G_i$,
rather than the original graph $G$.

\paragraph{Randomized construction.}
The randomized construction is simple.
We start with the low-diameter clustering of Miller, Peng, and
Xu~\cite{MPX13SPAA}, choosing its parameter so that the expected number
of inter-cluster edges is at most $2\varepsilon\beta|E|$.
We then repeatedly remove from each cluster any vertex that has fewer than
$\left(\frac12-\varepsilon\right)$ of its neighbors inside the cluster.
A counting argument shows that this trimming increases the number of
inter-cluster edges by a factor of at most $1/(2\varepsilon)$.
Hence, after trimming, at most $\beta|E|$ edges are inter-cluster in expectation.

Taking any constant $\beta \in (0,1)$ and repeating on the remaining edges gives
$O(\log n)$ layers with high probability.
Since each clustering takes $O(\log n)$ rounds and has weak diameter
$O(\log n)$, we obtain the following result.

\begin{restatable}[Randomized layered introvert decomposition] {theorem}{randintrovertdecomp} \label{thm:introvert-decomposition-rand} For every constant $\varepsilon \in (0,1/2)$, an $(\varepsilon,O(\log n),O(\log n))$-layered introvert network decomposition can be constructed in $O(\log^2 n)$ rounds with high probability. \end{restatable}

\paragraph{Deterministic construction.}
We next explain the main idea behind our deterministic construction.
A standard network decomposition partitions the vertices into
low-diameter clusters and assigns colors to the clusters so that
adjacent clusters receive different colors.
Ghaffari and Grunau~\cite{ghaffari_grunau_focs24} recently gave a
deterministic algorithm that constructs such a decomposition with
$O(\log n)$ colors and $O(\log n)$ cluster diameter in
$\widetilde O(\log^2 n)$ rounds.

Their result does not directly give what we need.
Our decomposition assigns \emph{edges} to layers rather than vertices
to colors, and every clustered vertex must additionally satisfy the
introvert property.
We therefore adapt their recursive algorithm in a white-box manner.

At the base case, we strengthen their low-diameter clustering
procedure so that an arbitrarily small constant fraction of the edges
under consideration remain inter-cluster.
This gives enough slack to apply our trimming procedure while still
making constant-factor progress.

A second modification is needed because our introvert condition is
defined with respect to all edges that remain when a layer is formed.
In an ordinary network decomposition, once some vertices are deferred
to a later recursive call, they do not affect whether the clusters
formed by the current call are valid.
Here, however, an edge deferred to a later recursive call is still
present and can contribute to the degree of a vertex in a cluster
formed earlier.
We therefore keep track of all \emph{active edges}, namely, edges that
have not yet been assigned to any layer, and require every introvert
condition to hold with respect to the full active edge set.
Accordingly, we strengthen the invariant of the recursive calls to
also bound the number of active edges lying outside the set currently
handled by the recursion.
We show that the recursive construction can be carried out while
preserving this strengthened invariant.

Finally, since the objects handled by the recursive algorithm are
vertices whereas our layers consist of edges, we run the recursion on
the line graph.
A clustering of the line graph into pairwise non-adjacent clusters
translates naturally into a clustering of the vertices of
$G$, with essentially the same diameter, while every inter-cluster edge corresponds to an unclustered vertex in the line graph.

\begin{restatable}[Deterministic layered introvert decomposition] {theorem}{detintrovertdecomp} \label{thm:introvert-decomposition-det} For every constant $\varepsilon \in (0,1/2)$, an $(\varepsilon,O(\log n),O(\log n))$-layered introvert network decomposition can be constructed in $\widetilde O(\log^2 n)$ rounds deterministically. \end{restatable}

\paragraph{Why being introvert helps.}
A standard network decomposition is particularly useful for turning
sequential local algorithms into distributed ones.
Roughly speaking, one processes the color classes of the decomposition
one at a time; within each color class, the low-diameter clusters can be
processed independently and in parallel.
For example, this immediately parallelizes the familiar greedy algorithms
for maximal independent set, maximal matching, $(\Delta+1)$-vertex coloring, and $(2\Delta-1)$-edge coloring.
More generally, Ghaffari, Kuhn, and Maus~\cite{GhaffariKM17} formalized
this connection through the $\mathsf{SLOCAL}$ model, which captures
algorithms that process the vertices sequentially while making each
decision using only a local neighborhood.
Network decomposition provides a general mechanism for translating such
sequential locality into distributed locality, and has consequently
become a central tool for obtaining efficient deterministic distributed
algorithms~\cite{GhaffariKM17,RozhonGhaffari20}.

Introvert decomposition gives us something extra: its clusters not only
have low diameter, but also keep many of their neighbors close.
This seemingly modest additional guarantee turns out to make a
surprisingly large difference.
The introvert guarantee allows us to process clusters even in settings
where an arbitrary partial solution may leave too little room to extend
the solution greedily.

List edge coloring below the $2\Delta-1$ greedy threshold is one example.
An arbitrary partial coloring may leave an uncolored edge with no
available color.
With introvert clustering, however, roughly half of the relevant
edges at each vertex remain inside the current cluster.
This leaves precisely the additional slack needed to make lists of size
$\left(\frac{3}{2}+\varepsilon\right)\Delta$ sufficient.
For locally balanced cut, the same phenomenon appears in a different
form: roughly half of each vertex's relevant neighbors remain inside its
cluster, and computing a locally maximum cut within the cluster guarantees
that at least half of those neighbors lie on the opposite side.
The two factors of $1/2$ lead directly to the $1/4$ guarantee.

A particularly appealing aspect of our framework is its simplicity.
Once the layered introvert network decomposition from
\Cref{thm:introvert-decomposition-rand,thm:introvert-decomposition-det}
is available, both applications admit remarkably simple algorithms:
the edge coloring algorithm processes the layers in reverse order and
colors each cluster, while the cut algorithm processes them in forward
order and computes a locally maximum cut inside each cluster.
Thus, beyond the new quantitative bounds in
\Cref{thm:intro-main,thm:intro-cut}, these applications illustrate the
extra algorithmic power provided by the introvert guarantee, opening up
new opportunities for using network decomposition beyond the usual
greedy regime.

\paragraph{The $\mathbf{1/2}$ barrier.}
The threshold $1/2$ in the introvert property is essentially \emph{tight}.
Consider the graph in \Cref{fig:introvert-half-barrier}, consisting of
$t=2n/\Delta$ consecutive layers
$L_1,\ldots,L_t$, each with $\Delta/2$ vertices, where every pair of
consecutive layers induces a complete bipartite graph.
Thus, every vertex in an internal layer has degree $\Delta$.

Suppose that a nonempty cluster $C$ satisfies
$|N(v)\cap C|>\deg(v)/2$ for every $v\in C$.
Let $L_i$ be the first layer intersecting $C$.
If $i>1$, then any $v\in C\cap L_i$ has no neighbor in $C$ in
$L_{i-1}$, and hence at most its $\Delta/2$ neighbors in $L_{i+1}$
can belong to $C$.
Thus $|N(v)\cap C|\leq\Delta/2=\deg(v)/2$, a contradiction.
Therefore $C$ must intersect $L_1$.
By the same argument from the other end, $C$ must also intersect
$L_t$; in fact, it must intersect every layer.

Hence every nonempty cluster satisfying the strict $1/2$ introvert
condition has weak diameter $\Omega(t)=\Omega(n/\Delta)$.
In particular, any clustering with this stronger guarantee
must either leave every vertex unclustered or contain a cluster of
diameter $\Omega(n/\Delta)$.

\usetikzlibrary{decorations.pathreplacing}

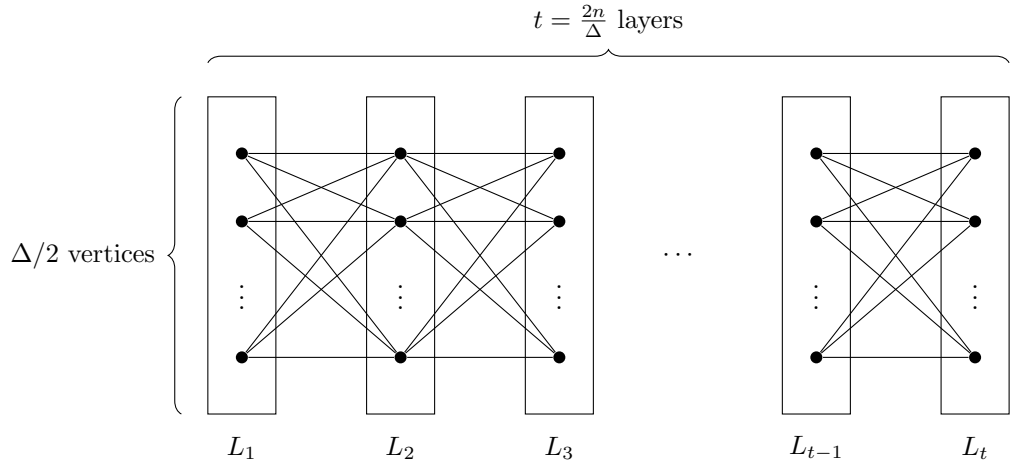
\begin{figure}[ht!]
\centering
\begin{tikzpicture}[
    x=1cm,
    y=1cm,
    every node/.style={font=\small},
    vtx/.style={circle, fill=black, inner sep=1.6pt}
]

\newcommand{\drawlayer}[3]{%
    \draw (#2,0) rectangle ++(0.9,4.2);
    \node at ({#2+0.45},-0.45) {$#3$};
    \node[vtx] (#1a) at ({#2+0.45},3.45) {};
    \node[vtx] (#1b) at ({#2+0.45},2.55) {};
    \node at ({#2+0.45},1.65) {$\vdots$};
    \node[vtx] (#1c) at ({#2+0.45},0.75) {};
}

\newcommand{\drawcompletebip}[2]{%
    \foreach \u in {a,b,c}{
        \foreach \v in {a,b,c}{
            \draw[thin] (#1\u) -- (#2\v);
        }
    }
}

\drawlayer{A}{0.0}{L_1}
\drawlayer{B}{2.1}{L_2}
\drawlayer{C}{4.2}{L_3}
\drawlayer{D}{7.6}{L_{t-1}}
\drawlayer{E}{9.7}{L_t}

\drawcompletebip{A}{B}
\drawcompletebip{B}{C}
\drawcompletebip{D}{E}

\node at (6.25,2.1) {$\cdots$};

\draw[decorate,decoration={brace,amplitude=5pt}]
    (-0.35,0) -- (-0.35,4.2)
    node[midway,left=6pt] {$\Delta/2$ vertices};

\draw[decorate,decoration={brace,amplitude=5pt}]
    (0,4.65) -- (10.6,4.65)
    node[midway,above=6pt] {$t=\frac{2n}{\Delta}$ layers};

\end{tikzpicture}
\caption{A graph illustrating the $1/2$ barrier for introvert clustering.}
\label{fig:introvert-half-barrier}
\end{figure}

\subsection{Roadmap}

The remainder of the paper is organized as follows.
In \Cref{sec: list coloring}, we present our list edge coloring algorithm.
In \Cref{sec: approx locally max cut}, we give our algorithm for locally
balanced cut.
We then present the randomized construction of layered introvert network
decomposition in \Cref{sec: introvert clustering}.
We conclude with a discussion of open questions and future directions in
\Cref{sec: conclusion}.
The deterministic construction of layered introvert network
decomposition, based on an adaptation of the recursive network
decomposition algorithm of Ghaffari and Grunau~\cite{ghaffari_grunau_focs24}, is deferred to
\Cref{app:det-introvert,app:small-loss-proof}.

\section{List Edge Coloring}
\label{sec: list coloring}

In this section, we show how a layered introvert network decomposition
gives our list edge coloring algorithms.
Once the decomposition is available, the algorithm is simple: we process
the layers in reverse order and color all clusters of the current layer
in parallel.
The main point is to show that the introvert property leaves sufficiently
many colors available when a cluster is processed.

\subsection{List Edge Coloring Inside a Cluster}

We begin by reviewing the known existential list edge coloring results that we use to
color the edges within each cluster.
A \emph{list assignment} $L$ assigns to each edge $e$ a set $L(e)$ of
available colors.
A \emph{proper $L$-edge coloring} assigns to each edge $e$ a color from
$L(e)$ such that any two edges sharing an endpoint receive distinct
colors.

The first result gives a
local list size condition for general graphs.

\begin{theorem}[Bonamy, Delcourt, Lang, and Postle~\cite{BonamyDLP24}]
\label{thm:local-list-coloring}
For every constant $\gamma>0$, there exists a constant
$\Delta_0(\gamma)$ such that the following holds.
Let $H$ be a graph with maximum degree
$\Delta(H)\geq\Delta_0(\gamma)$ and minimum degree
$\delta(H)\geq\ln^{25}\Delta(H)$.
If $L$ is a list assignment satisfying
$$
|L(\{u,v\})|
\geq
(1+\gamma)\max\{\deg_H(u),\deg_H(v)\}
$$
for every edge $\{u,v\}\in E(H)$, then $H$ admits a proper
$L$-edge coloring.
\end{theorem}

For bipartite graphs, neither the minimum degree assumption nor the
multiplicative slack is needed.
\begin{theorem}[Borodin, Kostochka, and Woodall~\cite{BorodinKW97}]
\label{thm:bipartite-local-list-coloring}
Let $H$ be a bipartite graph with a list assignment $L$.
If
$$
|L(\{u,v\})|
\geq
\max\{\deg_H(u),\deg_H(v)\}
$$
for every edge $\{u,v\}\in E(H)$, then $H$ admits a proper
$L$-edge coloring.
\end{theorem}

For general graphs, the minimum degree assumption in
\Cref{thm:local-list-coloring} is inconvenient for our application,
since the subgraph induced by a cluster may contain vertices of small
degree.
The following simple corollary removes this assumption at the cost of
requiring an absolute lower bound on the list sizes.
\begin{corollary}
\label{cor:local-list-coloring-no-min-degree}
For every constant $\gamma>0$, there exists a constant
$\Delta_0(\gamma)$ such that the following holds for every integer
$\Delta\geq\Delta_0(\gamma)$.
Let $H$ be a graph with maximum degree at most $\Delta$.
If $L$ is a list assignment satisfying
$$
|L(\{u,v\})|
\geq
(1+\gamma)
\max\{\deg_H(u),\deg_H(v),\lceil\ln^{25}\Delta\rceil\}
$$
for every edge $\{u,v\}\in E(H)$, then $H$ admits a proper
$L$-edge coloring.
\end{corollary}

\begin{proof}
We may assume that $\Delta$ is sufficiently large that
$\lceil\ln^{25}\Delta\rceil<\Delta$.
Set
$
d_0=\lceil\ln^{25}\Delta\rceil
$.
We augment $H$ to a graph $H^+$ with minimum degree at least $d_0$,
while preserving the degrees of vertices of $H$ that already have
degree at least $d_0$.

For each vertex $v$ with $\deg_H(v)<d_0$, take a fresh copy $Q_v$ of a $(d_0+1)$-clique  and connect $v$ to exactly
$d_0-\deg_H(v)$ distinct vertices of $Q_v$.
Then $v$ has degree exactly $d_0$ in $H^+$, while every vertex of
$Q_v$ has degree either $d_0$ or $d_0+1$.
Vertices of degree at least $d_0$ in $H$ are left unchanged.

Finally, add a disjoint copy of a $(\Delta+1)$-clique.
Since $d_0+1\leq\Delta$, the resulting graph satisfies
$$
\Delta(H^+)=\Delta
\qquad\text{and}\qquad
\delta(H^+)\geq d_0\geq\ln^{25}\Delta.
$$
Moreover, every original vertex $v\in V(H)$ satisfies
$$
\deg_{H^+}(v)
=
\max\{\deg_H(v),d_0\}.
$$
Hence, for every original edge $\{u,v\}\in E(H)$,
\[|L(\{u,v\})| \geq (1+\gamma)\max\{\deg_H(u),\deg_H(v),d_0\} = (1+\gamma)\max\{\deg_{H^+}(u),\deg_{H^+}(v)\}.\]
Extend $L$ to the newly added edges by assigning them sufficiently large
lists of colors.
The resulting list assignment on $H^+$ satisfies the conditions of
\Cref{thm:local-list-coloring}, so $H^+$ has a proper list edge
coloring.
Restricting this coloring to the original edges of $H$ yields a proper
$L$-edge coloring of $H$.
\end{proof}
 
\subsection{The Distributed Coloring Algorithm}

We now prove our list edge coloring theorem.

\mainthm*

\begin{proof}
Fix the constant $\varepsilon>0$ from the theorem.
It suffices to consider $0<\varepsilon\leq 1/2$, since a result for a
smaller value of $\varepsilon$ also implies the result for any larger
one.
Set $\eta=\varepsilon/2$ and $\gamma=\varepsilon/2$.

\paragraph{The algorithm.}
We first construct an
$(\eta,\ell,D)$-layered introvert network decomposition
$E=E_1\cup\cdots\cup E_\ell$, with clusterings
$\mathcal C_1,\ldots,\mathcal C_\ell$, where
$\ell=O(\log n)$ and $D=O(\log n)$.
Recall that
$G_i=(V,E_i\cup E_{i+1}\cup\cdots\cup E_\ell)$.

We process the layers in reverse order,
$E_\ell,E_{\ell-1},\ldots,E_1$.
Suppose that all edges in
$E_{i+1}\cup\cdots\cup E_\ell$ have already been colored.
For each cluster $C\in\mathcal C_i$, we color the edges of $G_i[C]$,
which are precisely the edges of $E_i$ with both endpoints in $C$.
All clusters in the same layer are processed in parallel.

\paragraph{How many colors remain?}
Fix a cluster $C\in\mathcal C_i$.
For every vertex $v\in C$, the introvert property gives
$\deg_{G_i[C]}(v)\geq
(\frac12-\eta)\deg_{G_i}(v)$.
It follows that
\begin{equation}
\label{eq:introvert-colored-degree}
\deg_{G_i}(v)-\deg_{G_i[C]}(v)
\leq
\left(\frac12+\eta\right)\deg_{G_i}(v)
\leq
\left(\frac12+\eta\right)\Delta.
\end{equation}
The left-hand side is exactly the number of edges incident to $v$ that
belong to later layers and have therefore already been colored.

Now consider an edge $e=\{u,v\}\in E(G_i[C])$. Without loss of generality, we assume that
\[\deg_{G_i[C]}(u)\geq\deg_{G_i[C]}(v).\]
Let $L_i(e)$ denote its residual list after removing all colors already
used by colored edges incident to $u$ or $v$.
The number of colors removed from $L(e)$ is at most
$$
\bigl(\deg_{G_i}(u)-\deg_{G_i[C]}(u)\bigr)
+
\bigl(\deg_{G_i}(v)-\deg_{G_i[C]}(v)\bigr).
$$
Therefore,
\begin{align}
|L_i(e)|
&\geq
\left(\frac32+\varepsilon\right)\Delta
-\bigl(\deg_{G_i}(u)-\deg_{G_i[C]}(u)\bigr)
-\left(\frac12+\eta\right)\Delta
&& \text{by \eqref{eq:introvert-colored-degree}} \notag\\
&=
\deg_{G_i[C]}(u)
+\bigl(\Delta-\deg_{G_i}(u)\bigr)
+(\varepsilon-\eta)\Delta \notag\\
&\geq
\deg_{G_i[C]}(u)+(\varepsilon-\eta)\Delta
&& \text{as $\deg_{G_i}(u)\leq\Delta$} \notag\\
&\geq
(1+\varepsilon-\eta)\deg_{G_i[C]}(u)
&& \text{as $\deg_{G_i[C]}(u)\leq\Delta$}. \notag
\end{align}
Since $\gamma=\varepsilon-\eta=\varepsilon/2$ and
$\deg_{G_i[C]}(u)=
\max\{\deg_{G_i[C]}(u),\deg_{G_i[C]}(v)\}$, we obtain
\begin{equation}
\label{eq:residual-local-degree}
|L_i(\{u,v\})|
\geq
(1+\gamma)
\max\{\deg_{G_i[C]}(u),\deg_{G_i[C]}(v)\}.
\end{equation}

\paragraph{Handling small degrees.}
For general graphs, we also need the absolute list size guarantee
required by \Cref{cor:local-list-coloring-no-min-degree}.
Applying \eqref{eq:introvert-colored-degree} to both endpoints of an
edge $e\in E(G_i[C])$ gives
\begin{align*}
|L_i(e)|
&\geq
\left(\frac32+\varepsilon\right)\Delta
-(1+2\eta)\Delta =
\left(\frac12+\varepsilon-2\eta\right)\Delta
=
\frac{\Delta}{2}.
\end{align*}
For sufficiently large $\Delta$, depending only on $\varepsilon$,
$$
\frac{\Delta}{2}
\geq
(1+\gamma)\lceil\ln^{25}\Delta\rceil.
$$
Together with \eqref{eq:residual-local-degree}, this implies that every
edge $\{u,v\}\in E(G_i[C])$ satisfies
$$
|L_i(\{u,v\})|
\geq
(1+\gamma)
\max\left\{
\deg_{G_i[C]}(u),
\deg_{G_i[C]}(v),
\lceil\ln^{25}\Delta\rceil
\right\}.
$$
Therefore, by \Cref{cor:local-list-coloring-no-min-degree},
$G_i[C]$ admits a proper list edge coloring from the residual lists.

Since each residual list excludes all colors already used by incident
edges in later layers, this coloring extends the existing partial
coloring without introducing any conflict.

\paragraph{Distributed implementation.}
The clusters in $\mathcal C_i$ are vertex-disjoint, so all clusters of
the same layer can be processed simultaneously.
Each cluster has weak diameter at most $D$ in $G$.
Since messages in the $\LOCAL$ model have unbounded size, all
information about $G_i[C]$ and its residual lists can be gathered in
$O(D)$ rounds.
Local computation is unrestricted, so a proper list edge coloring whose
existence is guaranteed above can then be computed locally, and the
resulting colors can be communicated in another $O(D)$ rounds.

Thus, given the layered decomposition, each layer can be processed in
$O(D)$ rounds, and all layers can be colored in
$O(\ell D)=O(\log^2 n)$ rounds.

Using \Cref{thm:introvert-decomposition-rand}, we obtain an
$O(\log^2 n)$-round randomized algorithm with high probability.
Using \Cref{thm:introvert-decomposition-det}, we obtain a
$\widetilde O(\log^2 n)$-round deterministic algorithm.

\paragraph{Bipartite graphs.}
Suppose now $G$ is bipartite, so every graph $G_i[C]$ is also bipartite.
By \eqref{eq:residual-local-degree},
$$
|L_i(\{u,v\})|
\geq
\max\{\deg_{G_i[C]}(u),\deg_{G_i[C]}(v)\}
$$
for every edge $\{u,v\}\in E(G_i[C])$.
Thus, by \Cref{thm:bipartite-local-list-coloring}, every cluster can be
colored without any lower bound on $\Delta$.
The same algorithm and round complexity bounds therefore apply for all $\Delta$ in
the bipartite case.
\end{proof}

\section{Locally Balanced Cut}
\label{sec: approx locally max cut}

In this section, we show how a layered introvert network decomposition
gives our locally balanced cut algorithm.
As in the edge coloring algorithm, we process the decomposition one
layer at a time, but here we process the layers in forward order.
When a vertex first appears in a cluster, we assign it to one side of
the cut, and its assignment is never changed afterward.

\approxcutthm*

\begin{proof}
Fix the constant $0<\varepsilon<1/4$ from the theorem and set
$\eta=2\varepsilon$.
We construct an
$(\eta,\ell,D)$-layered introvert network decomposition
$E=E_1\cup\cdots\cup E_\ell$, with clusterings
$\mathcal C_1,\ldots,\mathcal C_\ell$, where
$\ell=O(\log n)$ and $D=O(\log n)$.
Recall that
$G_i=(V,E_i\cup E_{i+1}\cup\cdots\cup E_\ell)$.

\paragraph{The algorithm.}
Initially, no vertex has been assigned to either side of the cut.
We process the layers in forward order,
$E_1,E_2,\ldots,E_\ell$.

Consider a cluster $C\in\mathcal C_i$.
Some vertices of $C$ may already have been assigned to a side of the
cut because they belonged to a cluster in an earlier layer.
We keep these assignments fixed and assign the remaining vertices of
$C$ so as to maximize the number of crossing edges in $G_i[C]$,
subject to the fixed assignments.
Once a vertex is assigned, its assignment is never changed.
All clusters in the same layer are processed in parallel.

Since every edge belongs to some layer, every vertex belongs to a
cluster in some layer and is therefore eventually assigned.

\paragraph{The guarantee for a newly assigned vertex.}
Suppose that $v$ is assigned for the first time while processing a
cluster $C\in\mathcal C_i$.
By the choice of the cut inside $C$, switching only $v$ to the other
side cannot increase the number of crossing edges in $G_i[C]$.
Therefore, at least half of the neighbors of $v$ in $G_i[C]$ are
assigned to the opposite side; otherwise, switching $v$ would strictly
increase the number of crossing edges.
Thus, $v$ has at least
$\frac12\deg_{G_i[C]}(v)$ neighbors in $C$ on the opposite side.

\paragraph{Relating to the original degree.}
Fix a vertex $v$, and let $i$ be the first layer in which $v$ belongs
to a cluster.
Let $C\in\mathcal C_i$ be the cluster containing $v$.
None of the edges incident to $v$ belongs to an earlier layer.
Indeed, if an incident edge belonged to $E_j$ for some $j<i$, then it
would be an intra-cluster edge of $\mathcal C_j$, implying that $v$
already belonged to a cluster in layer $j$.
Consequently,
\begin{equation}
\label{eq:first-layer-full-degree}
\deg_{G_i}(v)=\deg_G(v).
\end{equation}

By the introvert property and \eqref{eq:first-layer-full-degree},
\begin{equation*}
\deg_{G_i[C]}(v)
\geq
\left(\frac12-\eta\right)\deg_G(v).
\end{equation*}
Since $v$ is assigned for the first time when $C$ is processed, the
argument above shows that it has at least
$$
\frac12\deg_{G_i[C]}(v)
\geq
\frac12\left(\frac12-\eta\right)\deg_G(v)
=
\left(\frac14-\varepsilon\right)\deg_G(v)
$$
neighbors on the opposite side.
These assignments are never changed afterward, so the same guarantee
holds in the final cut.
Since this argument applies to every vertex, the resulting cut is
$\left(\frac14-\varepsilon\right)$-locally balanced.

\paragraph{Distributed implementation.}
As in the edge coloring algorithm, the clusters of each layer can be
processed in parallel, and the weak diameter bound allows each cluster
to gather all information needed to compute its cut in $O(D)$ rounds.
Thus, given the layered decomposition, all layers can be processed in
$O(\ell D)=O(\log^2 n)$ rounds.
Using \Cref{thm:introvert-decomposition-rand}, we obtain an
$O(\log^2 n)$-round randomized algorithm with high probability.
Using \Cref{thm:introvert-decomposition-det}, we obtain a
$\widetilde O(\log^2 n)$-round deterministic algorithm.
\end{proof}

\section{Randomized Construction of Introvert Decomposition}
\label{sec: introvert clustering}

In this section, we prove
\Cref{thm:introvert-decomposition-rand}.
The construction starts from the standard low-diameter clustering of
Miller, Peng, and Xu~\cite{MPX13SPAA}.
Such a clustering guarantees that few edges cross between clusters, but
this is only a global guarantee: an individual vertex may still have
most of its neighbors outside its own cluster.
We enforce the introvert property by a simple trimming procedure.

\subsection{Trimming Low-Diameter Clusters}

We first show that any clustering with few inter-cluster edges can be
trimmed so that every remaining clustered vertex satisfies the
introvert property, while increasing the number of inter-cluster edges
by only a constant factor.

\begin{lemma}[Trimming]
\label{lem:trimming}
Let $\mathcal C$ be a clustering of a graph $G=(V,E)$, and let
$U$ be the number of inter-cluster edges with respect to $\mathcal C$.
For every $0<\varepsilon<1/2$, there is a clustering
$\mathcal C'$ obtained by removing vertices from the clusters of
$\mathcal C$.
The clustering $\mathcal C'$ has the following properties.
\begin{itemize}
    \item Every $v\in C\in\mathcal C'$ satisfies
    $|N_G(v)\cap C|\geq
    \left(\frac12-\varepsilon\right)\deg_G(v)$.
    \item The number of inter-cluster edges with respect to
    $\mathcal C'$ is at most $U/(2\varepsilon)$.
\end{itemize}
\end{lemma}

\begin{proof}
Starting from $\mathcal C$, repeatedly remove any clustered vertex $v$
with fewer than
$\left(\frac12-\varepsilon\right)\deg_G(v)$ neighbors in its current
cluster.
When the process terminates, every remaining clustered vertex satisfies
the first property.

It remains to bound the number of inter-cluster edges created by the
trimming process.
For every vertex $v$ that is removed, consider the clustering
\ul{immediately before $v$ is removed}.
Let $d_{\mathrm{in}}(v)$ be the number of neighbors of $v$ that are
still in the same cluster as $v$ at this point, and let
$d_{\mathrm{out}}(v)=\deg_G(v)-d_{\mathrm{in}}(v)$.
Since $v$ is removed,
$d_{\mathrm{in}}(v)<
\left(\frac12-\varepsilon\right)\deg_G(v)$, and therefore
\begin{equation}
\label{eq:trimming-ratio}
d_{\mathrm{out}}(v)
>
\frac{1+2\varepsilon}{1-2\varepsilon}
d_{\mathrm{in}}(v).
\end{equation}

Let $D_{\mathrm{in}}$ and $D_{\mathrm{out}}$ denote the sums of
$d_{\mathrm{in}}(v)$ and $d_{\mathrm{out}}(v)$, respectively, over all
vertices removed during the process.
We bound $D_{\mathrm{out}}-D_{\mathrm{in}}$ by considering the
contribution of each edge.

First consider an edge that is initially intra-cluster.
If neither endpoint is removed, it contributes nothing.
If exactly one endpoint is removed, the edge is internal immediately
before that endpoint is removed, and hence contributes $-1$ to
$D_{\mathrm{out}}-D_{\mathrm{in}}$.
If both endpoints are removed, the edge is internal when the first
endpoint is removed and external when the second endpoint is removed,
so its two contributions cancel.
Thus, every initially intra-cluster edge contributes at most $0$.

Now consider an edge that is initially inter-cluster.
Since clusters only lose vertices, its endpoints can never enter the
same cluster.
Hence it contributes at most $1$ to $D_{\mathrm{out}}$ from each
endpoint, and therefore at most $2$ to
$D_{\mathrm{out}}-D_{\mathrm{in}}$.
Since there are $U$ initially inter-cluster edges, we obtain
\begin{equation}
\label{eq:trimming-difference}
D_{\mathrm{out}}-D_{\mathrm{in}}\leq 2U.
\end{equation}

Summing \eqref{eq:trimming-ratio} over all removed vertices gives
$D_{\mathrm{out}}>
\frac{1+2\varepsilon}{1-2\varepsilon}D_{\mathrm{in}}$.
Together with \eqref{eq:trimming-difference}, this yields
$$
\frac{4\varepsilon}{1-2\varepsilon}D_{\mathrm{in}}
<
D_{\mathrm{out}}-D_{\mathrm{in}}
\leq 2U,
$$
and hence
$$
D_{\mathrm{in}}
<
\frac{1-2\varepsilon}{2\varepsilon}U.
$$

Every initially intra-cluster edge that becomes inter-cluster is
counted once in $D_{\mathrm{in}}$, namely when its first endpoint is
removed.
Thus, the number of newly created inter-cluster edges is
$D_{\mathrm{in}}$.
The final number of inter-cluster edges is therefore at most
$$
U+D_{\mathrm{in}}
<
U+\frac{1-2\varepsilon}{2\varepsilon}U
=
\frac{U}{2\varepsilon},
$$
proving the second property.
\end{proof}

Since trimming only removes vertices from clusters, it cannot increase
their weak diameter.
More precisely, in \Cref{lem:trimming}, if every cluster of the original
clustering $\mathcal C$ has weak diameter at most $D$, then the same is
true for the trimmed clustering $\mathcal C'$, since every cluster of
$\mathcal C'$ is a subset of some cluster of $\mathcal C$.

We use the low-diameter clustering algorithm of Miller, Peng, and
Xu~\cite{MPX13SPAA}.
Although originally presented in the PRAM model, it is well known that
the algorithm admits an implementation in the $\LOCAL$ model with the
following guarantee; see, e.g.,~\cite{forster2022improved}.

\begin{theorem}[MPX clustering]
\label{thm:mpx-clustering}
For every $\lambda\in(0,1)$, there is a randomized
$O(\lambda^{-1}\log n)$-round $\LOCAL$ algorithm that computes a
clustering $\mathcal C$ of a graph $G=(V,E)$ such that
every cluster has strong diameter $O(\lambda^{-1}\log n)$ and the
expected number of inter-cluster edges is at most $\lambda|E|$.
\end{theorem}

Combining MPX low-diameter clustering with \Cref{lem:trimming} yields an
introvert clustering.
More precisely, we obtain exactly the guarantees of an
$(\varepsilon,\beta,O(\log n))$-introvert clustering, except that the
bound on the number of inter-cluster edges holds only in expectation.
To construct the layered decomposition, we therefore run this procedure
for a fixed $O(\log n)$ number of iterations and show that all edges have
been assigned with high probability.

\begin{lemma}[Randomized introvert clustering]
\label{lem:random-introvert-layer}
For all constants $0<\varepsilon<1/2$ and $\beta\in(0,1)$, there is
a randomized $O(\log n)$-round $\LOCAL$ algorithm that computes a clustering $\mathcal C$ of a graph $G=(V,E)$.
The clustering $\mathcal C$ satisfies the following properties.
\begin{itemize}
    \item Every cluster has weak diameter $O(\log n)$.
    \item Every $v\in C\in\mathcal C$ satisfies
    $|N_G(v)\cap C|\geq
    \left(\frac12-\varepsilon\right)\deg_G(v)$.
    \item If $U$ is the number of inter-cluster edges with respect to
    $\mathcal C$, then $\mathbb E[U]\leq\beta|E|$.
\end{itemize}
\end{lemma}

\begin{proof}
Apply \Cref{thm:mpx-clustering} with
$\lambda=2\varepsilon\beta$.
Let $U_0$ be the number of inter-cluster edges in the resulting
clustering.
Then $\mathbb E[U_0]\leq2\varepsilon\beta|E|$.

We now apply \Cref{lem:trimming}.
The resulting clustering satisfies the introvert property, and
if $U$ denotes its number of inter-cluster edges, then
$U\leq U_0/(2\varepsilon)$.
Therefore,
$\mathbb E[U]\leq\mathbb E[U_0]/(2\varepsilon)\leq\beta|E|$.

Since $\varepsilon$ and $\beta$ are constants, the MPX clustering takes
$O(\log n)$ rounds.
The trimming can also be implemented in $O(\log n)$ rounds.
Indeed, each MPX cluster has strong diameter $O(\log n)$ before
trimming, so one vertex can gather the induced subgraph of the cluster,
together with the degrees in $G$ of all its vertices, in
$O(\log n)$ rounds.
It can then simulate the trimming process locally and communicate the
resulting cluster membership back to the vertices.
Finally, by the observation following \Cref{lem:trimming}, trimming
preserves the $O(\log n)$ weak-diameter bound.
\end{proof}

\subsection{Constructing the Layered Decomposition}

We now repeatedly apply \Cref{lem:random-introvert-layer} to the edges
that remain to obtain the desired layered introvert decomposition.

\randintrovertdecomp*

\begin{proof}
Let $G=(V,E)$, and fix any constant $\beta\in(0,1)$, say
$\beta=1/2$.
We construct the decomposition one layer at a time.
Initially, let $G_1=G$.
At layer $i$, apply \Cref{lem:random-introvert-layer} to $G_i$ and let
$\mathcal C_i$ be the resulting clustering.
We define $E_i$ to be the set of intra-cluster edges of
$\mathcal C_i$ in $G_i$, and let $G_{i+1}$ consist of the remaining
inter-cluster edges.
We do this for $\ell=\lceil c\log n\rceil$ iterations, for a
sufficiently large constant $c$, to construct
$E_1,\ldots,E_\ell$ and $\mathcal C_1,\ldots,\mathcal C_\ell$.

By construction, every cluster of $\mathcal C_i$ has weak diameter
$O(\log n)$.
Moreover, the introvert condition holds: every
$v\in C\in\mathcal C_i$ satisfies
$|N_{G_i}(v)\cap C|\geq
\left(\frac12-\varepsilon\right)\deg_{G_i}(v)$.
Thus, to show that the resulting edge partition and clusterings form an
$(\varepsilon,\ell,O(\log n))$-layered introvert network decomposition,
it suffices to show that $E_1\cup\cdots\cup E_\ell=E$, which is
equivalent to $E(G_{\ell+1})=\emptyset$.

Let $M_i=|E(G_i)|$.
By \Cref{lem:random-introvert-layer}, conditioned on $G_i$,
\begin{equation}
\label{eq:remaining-edges-expectation}
\mathbb E[M_{i+1}\mid G_i]\leq\beta M_i.
\end{equation}
Taking expectations and iterating
\eqref{eq:remaining-edges-expectation} gives
$\mathbb E[M_i]\leq\beta^{i-1}|E|$.
Since $|E|<n^2$,
\[
\mathbb E[M_{\ell+1}]
\leq
\beta^\ell n^2
\leq
n^{2-c\log(1/\beta)}.
\]
As $M_{\ell+1}$ is a nonnegative integer, Markov's inequality gives
\[
\Pr[E_1\cup\cdots\cup E_\ell\neq E]
=
\Pr[M_{\ell+1}>0]
\leq
n^{2-c\log(1/\beta)}.
\]
Since $c$ can be chosen to be an arbitrarily large constant,
$E_1\cup\cdots\cup E_\ell=E$ with high probability, as desired.

Each layer takes $O(\log n)$ rounds by
\Cref{lem:random-introvert-layer}.
Therefore, the entire layered decomposition can be constructed in
$O(\log^2 n)$ randomized rounds with high probability.
\end{proof}

\section{Conclusions and Open Problems}
\label{sec: conclusion}

We introduced \emph{introvert clustering}, a low-diameter clustering
in which every clustered vertex keeps almost half of its neighbors
inside its own cluster.
This extra local guarantee allows network decompositions to be useful
beyond their standard role of parallelizing sequential local
algorithms: it gives enough structure to process clusters even when an
arbitrary partial solution may not be extendable greedily.
Using this idea, we obtained new algorithms for list edge
coloring and locally balanced cut.

For list edge coloring, Elkin, Pettie, and Su~\cite{EPS15} gave a
randomized list $(1+\varepsilon)\Delta$-edge coloring algorithm, which
can also be derandomized using general network decomposition
techniques.
Our algorithm uses more colors, but achieves
$\widetilde O(\log^2 n)$ deterministic rounds throughout the entire
range $\Delta\geq\Delta_0(\varepsilon)$, improving the resulting round
complexity when $\Delta$ is small.
Can one achieve the same $\widetilde O(\log^2 n)$ round complexity with
$(1+\varepsilon)\Delta$ colors? 

For locally balanced cut, our
$\left(\frac14-\varepsilon\right)$ guarantee can be obtained in
$\poly(\log n)$ rounds independently of $\Delta$, whereas the optimal
guarantee $1/2$ requires
$\Omega(\min\{\Delta,\sqrt n\})$ rounds~\cite{Balliu26Potential}.
What happens between these two regimes?
Can one compute a $\left(\frac12-\varepsilon\right)$-locally balanced cut in
$\poly(\log n)$ rounds independently of $\Delta$, or is there an
intermediate threshold beyond which a dependence on $\Delta$ is
unavoidable?

Finally, list edge coloring and locally balanced cut exploit the introvert property in quite different ways, suggesting that its usefulness may extend well beyond these two applications. Finding further applications of introvert clustering, and more broadly developing network decompositions with other useful local guarantees, is an exciting direction for future work.

 \section*{AI Disclosure}  
We used OpenAI's ChatGPT to assist with proofreading and improving the
clarity of the exposition, and to help identify potential issues in
proofs and citations throughout the manuscript.
The authors independently verified all mathematical arguments,
corrections, and references and take full responsibility for the
content of the paper.

\bibliography{references}
 
\appendix

\crefalias{section}{appendix}
\crefalias{subsection}{appendix}

\section{Deterministic Construction of Introvert Decomposition}
\label{app:det-introvert}

In this appendix, we prove the deterministic construction of layered
introvert decompositions
(\Cref{thm:introvert-decomposition-det}).
Our proof is a white-box adaptation of the deterministic network
decomposition algorithm of Ghaffari and
Grunau~\cite{ghaffari_grunau_focs24}.
We first develop the adaptation and prove the resulting deterministic
decomposition theorem.
The proof of a technical lemma needed in the base case is deferred to
\Cref{app:small-loss-proof}.
When referring to numbered lemmas,
theorems, and corollaries of Ghaffari and Grunau, we use the numbering
from \href{https://arxiv.org/abs/2410.19516v1}{arXiv:2410.19516v1}.

\subsection{High-Level Overview}

We begin with a high-level overview of the deterministic construction
and explain how it adapts the recursive network decomposition algorithm of Ghaffari and
Grunau~\cite{ghaffari_grunau_focs24}.
The main challenge is that a direct derandomization of our randomized
construction incurs an extra logarithmic factor, while the recursive
structure of their network decomposition algorithm allows us to avoid
this loss.
Our adaptation requires three additional ingredients: making the base
case introvert, keeping track of the active edges throughout the
recursion, and translating their vertex-based clustering framework to our
edge-based setting through the line graph.

\paragraph{Why a white-box adaptation is needed.}
Recall that our randomized construction
(\Cref{thm:introvert-decomposition-rand}) repeatedly applies an MPX
low-diameter clustering to the remaining edges and then trims the
resulting clusters.
Each application removes a constant fraction of the remaining edges,
so $O(\log n)$ applications suffice.
The known deterministic derandomizations of a single MPX-style
clustering require $\widetilde O(\log^2 n)$
rounds~\cite{GhaffariGHIR23,ghaffari_grunau_focs24}.
Thus, derandomizing these $O(\log n)$ applications separately would
give only a $\widetilde O(\log^3 n)$-round algorithm.

Ghaffari and Grunau~\cite{ghaffari_grunau_focs24} overcome this extra
logarithmic factor for standard network decomposition.
Recall that a network decomposition partitions the vertices into
$O(\log n)$ color classes such that every connected component induced
by one color class has diameter $O(\log n)$.
Rather than constructing the $O(\log n)$ color classes using
independent deterministic low-diameter clustering steps, their
algorithm constructs them together through a recursive procedure.
This reduces the overall complexity from
$\widetilde O(\log^3 n)$ to $\widetilde O(\log^2 n)$.
We adapt this recursive procedure directly.

\paragraph{1. Making the base case introvert.}
The clusters produced by the original construction need not satisfy
our introvert condition.
We modify the clustering procedure used at the base case so that each
clustering step leaves only an arbitrarily small constant fraction of
the edges under consideration inter-cluster.
This gives enough slack to apply trimming (\Cref{lem:trimming}).

\paragraph{2. Keeping track of active edges.}
During the recursive algorithm, some edges are deferred to later
recursive calls.
Nevertheless, until an edge is assigned to a layer, it must still
count toward the degree appearing in the introvert condition.
We therefore maintain throughout the recursion a set of
\emph{active edges}, consisting of all edges that have not yet been
assigned to a layer, and ensure that every introvert guarantee is
measured with respect to this active edge set.
This bookkeeping is needed to translate the recursive construction
into the layered decomposition of
\Cref{thm:introvert-decomposition-det}.

\paragraph{3. From vertex clustering to edge clustering.}
The algorithm of Ghaffari and Grunau decomposes vertices, whereas our
layered decomposition assigns edges to layers.
We bridge this difference using the line graph.
Throughout this appendix, let $\operatorname{Line}(G)$ denote the line
graph of $G$: its vertices are the edges of $G$, and two vertices are
adjacent if and only if the corresponding edges of $G$ share an
endpoint. Algorithms on $\operatorname{Line}(G)$ can be simulated on $G$ with
only a constant-factor overhead in the $\LOCAL$ model.

We make the correspondence precise.
Consider a clustering
$\mathcal F=\{F_1,\ldots,F_k\}$ of $\operatorname{Line}(G)$ such that
distinct clusters are non-adjacent.
Viewing each $F_j$ as a subset of $E(G)$, let
$C_j\subseteq V(G)$ be the set of endpoints of the edges in $F_j$, and
define
$\mathcal C=\{C_1,\ldots,C_k\}$.
Since $F_j$ and $F_{j'}$ are non-adjacent in the line graph for
$j\neq j'$, no edge in $F_j$ shares an endpoint with an edge in
$F_{j'}$.
Hence the sets $C_1,\ldots,C_k$ are pairwise disjoint, and
$\mathcal C$ is a clustering of $G$.
Furthermore, if every cluster $F_j$ has weak diameter at most $D$ in
$\operatorname{Line}(G)$, then every cluster $C_j$ has weak diameter at most
$D+1$ in $G$.

The correspondence also gives exactly the edge counting guarantee that
we need.
Let $A\subseteq E(G)$ be any set of edges, viewed equivalently as a
set of vertices of $\operatorname{Line}(G)$.
Suppose that at most $\eta|A|$ vertices of $A$ are not contained in
any of the clusters $F_1,\ldots,F_k$.

If the line-graph vertex corresponding to an edge $e\in A$ belongs to
some cluster $F_j$, then both endpoints of $e$ belong to $C_j$.
Hence $e$ is intra-cluster with respect to $\mathcal C$.
Therefore, every inter-cluster edge in $A$ must correspond to an
unclustered vertex of $\operatorname{Line}(G)$, and thus
\[
\#\{\text{inter-cluster edges in }A\}
\leq
\#\{\text{unclustered vertices of }A\text{ in }\operatorname{Line}(G)\}
\leq
\eta|A|.
\]
Thus, a low-diameter clustering that covers almost all vertices of the
line graph translates into a low-diameter clustering of $G$ with few
inter-cluster edges.

\paragraph{Roadmap.}
We first strengthen the base clustering procedure of Ghaffari and
Grunau so that the fraction of vertices left unclustered can be made
an arbitrarily small constant.
Via the line graph correspondence above, this gives a clustering of
$G$ with an arbitrarily small fraction of inter-cluster edges, to
which we can apply \Cref{lem:trimming}.
We then insert this modified base case into their recursive framework
and verify that, whenever a layer is formed, its introvert condition is
measured with respect to all edges that have not yet been assigned to
earlier layers.
This completes the proof of
\Cref{thm:introvert-decomposition-det}.

\subsection{Head Starts and Badness}

We review the basic ingredients of the recursive framework of
Ghaffari and Grunau~\cite{ghaffari_grunau_focs24}.

\paragraph{Head starts.}
The clustering procedures underlying the framework are based on the
same idea as MPX clustering.
One way to view MPX clustering is to imagine that every vertex
$v\in V(H)$ grows a BFS tree, with different vertices given different
\emph{head starts}.
If $h(v)$ is the head start of $v$, we may think of the BFS tree rooted
at $v$ as starting at time $-h(v)$.
Each vertex $u$ joins the cluster of a center whose BFS tree reaches
$u$ first, with ties broken arbitrarily.
Equivalently, $u$ chooses a center $v$ minimizing
$\dist_H(u,v)-h(v)$.
In the original randomized MPX clustering algorithm~\cite{MPX13SPAA},
the head starts are chosen randomly; in the deterministic framework of
Ghaffari and Grunau, they are constructed deterministically.

\paragraph{Badness.} For a fixed vertex $u$, several BFS trees may reach $u$ at the same
earliest time.
Controlling the number of such ties is useful because the 
clustering described above can then be made pairwise non-adjacent by
unclustering only a small number of vertices.
The notion of \emph{badness} measures, for each vertex $u$, the
maximum number of centers that can tie for the largest head start at
any distance up to a prescribed radius.

Formally, let $h:V(H)\to\mathbb N_{\geq 0}$ be a head start function.
For $u\in V(H)$ and an integer $r\geq 0$, let
\[
S_r^H(u)=\{v\in V(H):\dist_H(u,v)=r\}.
\]
For an integer $d\geq 0$, the \emph{badness} of $u$ with respect to
$h$ and $d$ is
\[
\bad_{h,d}^H(u)
=
\max_{r \, : \, \substack{0\leq r\leq d\\S_r^H(u)\neq\emptyset}}
\left|
\left\{
v\in S_r^H(u):
h(v)=\max_{w\in S_r^H(u)}h(w)
\right\}
\right|.
\]
When $H$ is clear from context, we omit the superscript.

To interpret this definition, fix a distance $r$ from $u$.
All vertices in $S_r^H(u)$ are equally far from $u$, so among
centers at distance $r$, those with the largest head start are
exactly those whose BFS trees reach $u$ first.
Thus, $\bad_{h,d}^H(u)$ measures the largest number of such tied
centers over all distances up to $d$.
When this quantity is small, the clustering described above can be
made pairwise non-adjacent by unclustering only a small fraction of the
vertices.
 
\paragraph{Recursively reducing the badness.}
The recursive algorithm repeatedly replaces the current head start
function by one with substantially smaller badness for almost all
vertices.
The following sampling lemma is the main tool for this reduction.

Throughout \Cref{app:det-introvert,app:small-loss-proof}, we write
$H=\operatorname{Line}(G)$, where $G$ is the original input graph, and
let $N$ be a known polynomial upper bound on $|V(H)|$.
We assume throughout that $N$ is at least a sufficiently large absolute
constant. This is without loss of generality, since we may replace $N$
by $\max\{N,N_0\}$ for any fixed constant $N_0$.

\begin{lemma}[Ghaffari--Grunau sampling lemma
{\cite[Corollary~5.2]{ghaffari_grunau_focs24}}]
\label{lem:gg-sampling}
There is an absolute constant $c$ for which the following holds.
Let $U\subseteq V(H)$, and consider the parameters
\begin{align*}
N_U &\geq |U|, &
d &= \left\lceil\log N\cdot(\log\log N)^c\right\rceil, &
B &\in[\log^5 N,N], &
B' &= B^{1/2+1/\log^2\log B}.
\end{align*}
Let
$h:V(H)\to\mathbb N_{\geq 0}$ be a head start function satisfying
$\bad_{h,d}^H(u)\leq B$ for every $u\in U$.

There is a deterministic distributed algorithm that, in
$\widetilde O(\log N\log B)$ rounds, computes a new head start
function $h':V(H)\to\mathbb N_{\geq 0}$ with the following properties.
\begin{itemize}
    \item
    $\max_{v\in V(H)}h'(v)
    \leq 2\max_{v\in V(H)}h(v)+1$.
    \item
    All but at most $N_U/B^3$ vertices $u\in U$ satisfy
    $\bad_{h',d}^H(u)\leq B'$.
\end{itemize}
\end{lemma}

Thus, except for at most $N_U/B^3$ vertices, one application of
\Cref{lem:gg-sampling} reduces the badness bound from $B$ to roughly
$\sqrt B$.
The recursive algorithm uses this reduction whenever $B$ is large,
recursing with the smaller parameter $B'$.
Once $B$ becomes polylogarithmic in $N$, the algorithm switches to a
base case procedure.

\paragraph{Distributed knowledge.}
In \Cref{lem:gg-sampling}, we assume that $N$, $N_U$, and $B$ are
known to all vertices, and that each vertex $v$ knows whether it belongs
to $U$ and knows its own head start $h(v)$.
We make the analogous assumptions throughout the rest of these
appendices: whenever a lemma or algorithm is given subsets of $V(H)$,
numerical parameters, or a head start function, every vertex $v$
knows the relevant parameters, its membership in the relevant subsets,
and its own head start.
We will omit these standard knowledge assumptions from subsequent
lemma statements.

\subsection{Small-Loss Clustering at the Base Case}
\label{app:small-loss-base}

In the base case of the recursion, where the badness is at
most $\log^{10}N$, Ghaffari and
Grunau~\cite[Lemma~5.3]{ghaffari_grunau_focs24} repeatedly invoke a
low-diameter clustering procedure and eventually leave only a
polylogarithmically small fraction of the vertices unresolved.
For our construction, we will use the same repetition.
However, each individual clustering step needs a stronger guarantee:
before we apply trimming, the fraction of edges left inter-cluster
must be an arbitrarily small constant.

The reason is the loss incurred by trimming.
Suppose that, after translating a clustering of
$\operatorname{Line}(G)$ back to $G$, at most a $\rho$-fraction of
some edge set is inter-cluster.
By \Cref{lem:trimming}, if at most a $\rho$-fraction of the
relevant edges are inter-cluster before trimming, then at most a
$\rho/(2\varepsilon)$-fraction remain afterward.
Thus, by taking $\rho$ to be a sufficiently small constant, we
obtain constant-factor progress.
In the proof for the introvert base case below, we use $\rho=\varepsilon/2$, which leaves at
most $1/4$ of these edges.

The following lemma gives the one-shot clustering primitive that we
need.

\begin{restatable}[Small-loss clustering]{lemma}{smalllossclustering}
\label{lem:small-loss-base}
There is an absolute constant $c$ for which the following holds.
Let $H$ be a graph, let $U\subseteq V(H)$, and let
\[
d=\left\lceil\log N\cdot(\log\log N)^c\right\rceil.
\]
Let $h:V(H)\to\mathbb N_{\geq 0}$ be a head start function satisfying
\[
\max_{v\in V(H)}h(v)\leq\log N
\qquad\text{and}\qquad
\bad_{h,d}^H(u)\leq\log^{10}N
\quad\text{for every }u\in U.
\]
For every constant $\rho>0$, there is a deterministic
$\widetilde O(\log N)$-round algorithm that computes a subset
$U'\subseteq U$ satisfying the following properties.
\begin{itemize}
    \item At most a $\rho$-fraction of $U$ is left unclustered:
    \(|U\setminus U'|\leq\rho|U|\).
    \item Every connected component of $H[U']$ has strong diameter
    $O(\log N)$.
\end{itemize}
\end{restatable}

We defer the proof of \Cref{lem:small-loss-base} to
\Cref{app:small-loss-proof}. 
\Cref{lem:small-loss-base} is a strengthened version of
\cite[Theorem~A.1]{ghaffari_grunau_focs24}.
Our only substantive change is to strengthen the guarantee on the
unclustered vertices of $U$: instead of allowing up to half of $U$ to
remain unclustered, we ensure that at most a $\rho$-fraction remains
unclustered, for any constant $\rho>0$.
The remaining differences in the statement simply match the interface
used in \cite[Lemma~5.3]{ghaffari_grunau_focs24}.

When $U\subseteq E(G)$ is viewed as a vertex set of
$\operatorname{Line}(G)$, the lemma therefore gives a
clustering of $G$ in which at most $\rho|U|$ edges of $U$ are
inter-cluster.
This is the one-shot guarantee that we will use in the base case.

\subsection{The Interface for Recursive Calls}
\label{app:interface}

Before describing the recursive algorithm, we first formalize the input to a recursive call and the invariants maintained throughout the recursion. This allows us to state both the base case and the recursive step more compactly.

\paragraph{The input to a recursive call.}
Since $H=\operatorname{Line}(G)$, we identify the vertices of $H$ with
the edges of $G$.
Every recursive call in our construction is specified by a tuple
\[
(A,U,N_U,B,h),
\]
with the following interpretation.
\begin{itemize}
    \item $A\subseteq E$ is the set of \emph{active edges}: the
    edges that have not yet been assigned to any earlier layer of the layered introvert network decomposition.
    \item $U\subseteq A$ is the set of edges currently handled by the
    recursive call.
    \item $N_U$ is an upper bound on $|U|$.
    \item $B$ is the current badness bound.
    \item $h:V(H)\to\mathbb N_{\geq0}$ is the current head start
    function.
\end{itemize}
We freely view $A$ and $U$ either as sets of edges of $G$ or, equivalently, as sets of vertices of $H=\operatorname{Line}(G)$, without explicitly distinguishing between the two interpretations.
Throughout the recursion, we maintain the following two invariants.
\begin{description}
    \item[Active set invariant.]
    The number of active edges outside the set currently handled by the
    recursive call satisfies
    \(|A\setminus U|\leq\frac{N_U}{B^3}\).

    \item[Badness invariant.]
    Every edge in $U$ satisfies
    \(\bad_{h,d}^H(e)\leq B\).
\end{description}
The badness invariant applies only to the edges in $U$.
However, the edges in $A\setminus U$ are still active, so they must
also be included when we check the introvert condition for any layer
constructed by the current call.
In addition to these two invariants, we will impose an upper bound on
$\max_{v\in V(H)}h(v)$; the appropriate bound differs between the base
case and the recursive step and will be stated separately for the two
cases.

\paragraph{Appending layers.}
It will be convenient to formalize the output produced by a recursive
call in a way that directly matches the definition of a layered
introvert network decomposition.

Let $A$ be the active edge set at the beginning of the call.
For $t\geq1$, an \emph{$(\varepsilon,D)$-valid block of $t$ layers
starting from $A$} consists of clusterings
$\mathcal C_1,\ldots,\mathcal C_t$ and edge sets
$F_1,\ldots,F_t$ defined as follows.
Set $A_1=A$, and for each $i\in[t]$, require the following.
\begin{description}
    \item[Low diameter.]
    $\mathcal C_i$ is a clustering of $(V,A_i)$, and every cluster of
    $\mathcal C_i$ has weak diameter at most $D$ in $G=(V,E)$.

    \item[Introvert property.]
    Every $v\in C\in\mathcal C_i$ satisfies
    $|N_{(V,A_i)}(v)\cap C|
    \geq
    \left(\frac12-\varepsilon\right)\deg_{(V,A_i)}(v)$.

    \item[Layer assignment.]
    $F_i$ is exactly the set of intra-cluster edges of
    $\mathcal C_i$ in $(V,A_i)$, and
    $A_{i+1}=A_i\setminus F_i$.
\end{description}

Thus, $A_i$ is precisely the set of edges that are active at the
beginning of the $i$th new layer, $F_i$ is the set of edges assigned
to that layer, and $A_{t+1}$ is the active edge set after the block
has been constructed.
These are exactly the three requirements in the definition of a
layered introvert network decomposition, restricted to a consecutive
block of layers.

Valid blocks can therefore be appended one after another: if one
block ends with active edge set $A'$, the next block simply starts
from $A'$.
Consequently, if we start with $A=E$ and keep appending valid blocks
until no active edge remains, then the resulting layer edge sets
partition $E$ and, together with their clusterings, form an
$(\varepsilon,\ell,D)$-layered introvert network decomposition, where
$\ell$ is the total number of layers constructed.

\subsection{The Base Case}
\label{app:introvert-base-case}

With the interface for recursive calls in place, we now turn to the base
case.
We combine \Cref{lem:small-loss-base} with trimming to obtain the base
case of our recursive construction.
This plays the same role as
\cite[Lemma~5.3]{ghaffari_grunau_focs24} in the recursion of Ghaffari
and Grunau.

\begin{lemma}[Base case]
\label{lem:introvert-base-case}
There is an absolute constant $c$ for which the following holds.
Consider a recursive call $(A,U,N_U,B,h)$ with
\[
d=
\left\lceil\log N\cdot(\log\log N)^c\right\rceil,
\qquad
B\in[\log^5 N,\log^{10}N], \qquad \max_{v\in V(H)}h(v)\leq\log N.
\]
For every constant $0<\varepsilon<1/2$, there is a deterministic
$\widetilde O(\log N)$-round algorithm that constructs an
$(\varepsilon,O(\log N))$-valid block of
$O(\log\log N)$ layers starting from $A$.
After this block, at most
\(
\frac{N_U}{B^2}
\) 
edges of $U$ remain active.
\end{lemma}

\begin{proof}
Let $A_1=A$ and $U_1=U$.
At the beginning of iteration $i$, let $A_i$ be the current active
edge set and let $U_i=U\cap A_i$.
Since $U_i\subseteq U$, we still have
$\bad_{h,d}^H(e)\leq B\leq\log^{10}N$ for every $e\in U_i$. 

\paragraph{One iteration.}
Apply \Cref{lem:small-loss-base} to $U_i$ with $\rho=\varepsilon/2$.
Translate the resulting connected components in the line graph into a
clustering $\mathcal C_i$ of $G$.
By \Cref{lem:small-loss-base}, every cluster of $\mathcal C_i$ has weak diameter $O(\log N)$ in $G$.

An active edge can be inter-cluster with respect to
$\mathcal C_i$ only if it lies outside $U_i$, or if it belongs to
$U_i$ but is among the at most $(\varepsilon/2)|U_i|$ edges left
unclustered by \Cref{lem:small-loss-base}.
Since
$A_i\setminus U_i\subseteq A\setminus U$, the current number
of active inter-cluster edges is at most
\[
|A\setminus U|+\frac{\varepsilon}{2}|U_i|.
\]

Now we do the trimming step. Apply \Cref{lem:trimming} to $\mathcal C_i$ in the active graph
$(V,A_i)$, and let $\widehat{\mathcal C}_i$ be the resulting
clustering.
Every clustered vertex now satisfies the
$(1/2-\varepsilon)$-introvert condition with respect to $A_i$, and
trimming does not increase the weak diameter.
Let $F_i$ be the set of intra-cluster edges of
$\widehat{\mathcal C}_i$ in $(V,A_i)$, assign $F_i$ to the current
layer, and set $A_{i+1}=A_i\setminus F_i$.
Thus, $A_{i+1}$ is exactly the set of active inter-cluster edges after
trimming.

By \Cref{lem:trimming},
\[
|A_{i+1}|
\leq
\frac{1}{2\varepsilon}
\left(
|A\setminus U|+\frac{\varepsilon}{2}|U_i|
\right)
=
\frac{|A\setminus U|}{2\varepsilon}
+\frac14|U_i|.
\]
Since $U_{i+1}\subseteq A_{i+1}$, we obtain the recurrence
\[
|U_{i+1}|
\leq
\frac{|A\setminus U|}{2\varepsilon}
+\frac14|U_i|.
\]
\paragraph{Residual active edges.} Iterating the recurrence for
$t=\lceil\log_4(2B^2)\rceil =O(\log\log N)$ iterations: 
\begin{align*}
|U_{t+1}|
&\leq
4^{-t}|U|
+
\frac{|A\setminus U|}{2\varepsilon}
\sum_{j=0}^{t-1}4^{-j}
&& \text{by repeatedly applying the recurrence} \\
&\leq
4^{-t}|U|
+
\frac{2}{3\varepsilon}|A\setminus U|
&& \text{since $\sum_{j=0}^{t-1}4^{-j}\leq 4/3$} \\
&\leq
4^{-t}N_U
+
\frac{2N_U}{3\varepsilon B^3}
&& \text{by $|U|\leq N_U$ and the active set invariant
$|A\setminus U|\leq N_U/B^3$} \\
&\leq
\frac{N_U}{2B^2}
+
\frac{2N_U}{3\varepsilon B^3}
&& \text{since $t=\lceil\log_4(2B^2)\rceil$} \\
&\leq
\frac{N_U}{2B^2}
+
\frac{N_U}{2B^2}
&& \text{since $B\geq\log^5N = \omega(1)$} \\
&=
\frac{N_U}{B^2}.  
\end{align*}
Therefore, after this block, at most
\(
\frac{N_U}{B^2}
\) 
edges of $U$ remain active.

\paragraph{Validity and round complexity.}
By construction, each iteration produces one layer satisfying the
three conditions of an
$(\varepsilon,O(\log N))$-valid block: the clusters have weak
diameter $O(\log N)$, the introvert property is measured with respect
to the active edge set $A_i$, and exactly the active intra-cluster
edges are assigned to the layer.
Therefore, the $t=O(\log\log N)$ layers form an
$(\varepsilon,O(\log N))$-valid block starting from $A$.

Each iteration applies \Cref{lem:small-loss-base} followed by the
trimming step of \Cref{lem:trimming}, and takes
$\widetilde O(\log N)$ rounds.
Since there are $O(\log\log N)$ iterations, the total round complexity
is $\widetilde O(\log N)$.
\end{proof}

\subsection{The Recursive Step}
\label{app:recursive-introvert}

We now describe the recursive step of our construction.
It enhances 
\cite[Lemma~5.4]{ghaffari_grunau_focs24} by incorporating the invariant about active edges \(|A\setminus U|\leq\frac{N_U}{B^3}\) introduced above.

For convenience, define
\[
\Phi(B)
=
\left(4.5-\frac{2\log B}{\log N}\right)
\left(1+\frac{100}{\log\log B}\right)
\frac{\log N}{\log B}.
\]
This is the upper bound on the maximum head start maintained in
Algorithm~1 and the proof of Lemma~5.4 of Ghaffari and Grunau.
Its precise form will not be important to us.
For
\[
B'=B^{1/2+1/\log^2\log B},
\]
their analysis gives
\begin{equation}
\label{eq:phi-recursion}
2\Phi(B)+1\leq\Phi(B')
\end{equation}
whenever $B$ is larger than a universal constant. Since throughout our recursion $B\geq\log^5 N = \omega(1)$, the condition holds.
This inequality is verified in the proof of
\cite[Lemma~5.4]{ghaffari_grunau_focs24}; we will use it without
repeating the calculation.

\paragraph{Structure of the recursive step.}
When $B$ is larger than the range handled by the base case, the
sampling lemma first reduces the badness bound from $B$ to
$B'=B^{1/2+1/\log^2\log B}$ for all but at most $N_U/B^3$ edges of
$U$.
We then make two recursive calls with parameter $B'$.
The first call handles the edges satisfying this improved badness
bound, and the second call handles the subset of these edges that
remain active after the first call, using a smaller value of $N_U$.

At first sight, the active set invariant \(|A\setminus U|\leq\frac{N_U}{B^3}\) may seem difficult to
preserve in the second call: its value of $N_U$ is smaller, while we
cannot rely on the number of active edges outside the set $U$ handled
by the call to decrease.
The key point is that $B$ also decreases, from $B$ to roughly
$B'=B^{1/2+1/\log^2\log B}$.
Since the active set invariant allows $N_U/B^3$ such edges, this
decrease in $B$ creates enough additional slack to compensate for the
smaller value of $N_U$.
The proof below verifies this formally.

\begin{lemma}[Recursive layered introvert decomposition]
\label{lem:recursive-introvert}
There is an absolute constant $c$ for which the following holds.
Consider a recursive call $(A,U,N_U,B,h)$ with
\[
d=\left\lceil\log N\cdot(\log\log N)^c\right\rceil,
\qquad
B\in[\log^5 N,N], \qquad \max_{v\in V(H)} h(v)\leq \Phi(B).
\]
For every constant $0<\varepsilon<1/2$, there is a deterministic
$\widetilde O(\log N\log B)$-round algorithm that constructs an
$(\varepsilon,O(\log N))$-valid block of $O(\log B)$ layers starting
from $A$.
After this block, at most
\(
\frac{N_U}{B^2}
\)
edges of $U$ remain active.
\end{lemma}

\begin{proof}
We follow the recursion of
\cite[Lemma~5.4]{ghaffari_grunau_focs24}, while additionally carrying
the active edge set $A$ through every recursive call.

\paragraph{Base case.}
Suppose that $B\leq\log^{10}N$.
Since $\Phi(B)\leq\log N$, the assumptions of \Cref{lem:introvert-base-case} are
satisfied, and that lemma directly gives an
$(\varepsilon,O(\log N))$-valid block of $O(\log\log N)$ layers,
leaving at most $N_U/B^2$ edges of $U$ active.
Since $B\geq\log^5N$, we have
$O(\log\log N)=O(\log B)$, as required.

\paragraph{Inductive step.}
Suppose now that $B>\log^{10}N$.
Set $B' = B^{1/2+1/\log^2\log B}$.
Apply \Cref{lem:gg-sampling} to $U$.
Let $h'$ be the resulting head start function, and define
\[
U^{(1)}
=
\{e\in U:\bad_{h',d}^H(e)\leq B'\}.
\]
By \Cref{lem:gg-sampling},
\[
|U\setminus U^{(1)}|\leq\frac{N_U}{B^3}.
\]
Moreover, by \Cref{lem:gg-sampling} and \eqref{eq:phi-recursion},
\[
\max_{v\in V(H)}h'(v)
\leq
2\max_{v\in V(H)}h(v)+1
\leq
2\Phi(B)+1
\leq
\Phi(B'),
\]
Thus $h'$ satisfies the head start requirement for recursive calls
with parameter $B'$.

\paragraph{First recursive call.}
We first recurse on $U^{(1)}$, keeping the full active edge set $A$.
The active set invariant is preserved because
\[
|A\setminus U^{(1)}|
=
|A\setminus U|+|U\setminus U^{(1)}|
\leq
\frac{2N_U}{B^3}
\leq
\frac{N_U}{(B')^3},
\]
where the last inequality holds because $B' = o(B)$.
The badness invariant holds by the definition of $U^{(1)}$.
Hence the induction hypothesis applies to
$(A,U^{(1)},N_U,B',h')$.

Let $A^{(2)}$ be the active edge set after this first recursive call,
and set
\[
U^{(2)}=U^{(1)}\cap A^{(2)}.
\]
By the induction hypothesis,
\(
|U^{(2)}|\leq\frac{N_U}{(B')^2}
\).
We therefore define
\[
N_U^{(2)}=\frac{N_U}{(B')^2}.
\]

\paragraph{Second recursive call.}
We now recurse on $U^{(2)}$ with active edge set $A^{(2)}$.
An edge of $A^{(2)}\setminus U^{(2)}$ must already have been outside
$U^{(1)}$ before the first recursive call.
Therefore, the active set
invariant is satisfied:
\[
|A^{(2)}\setminus U^{(2)}|
\leq
|A\setminus U^{(1)}|
\leq
\frac{2N_U}{B^3}
\leq
\frac{N_U}{(B')^5}
=
\frac{N_U^{(2)}}{(B')^3},
\]
where the inequality $\frac{2N_U}{B^3}
\leq
\frac{N_U}{(B')^5}$ is due to $(B')^5 \in o(B^3)$.
Since $U^{(2)}\subseteq U^{(1)}$, the badness invariant remains
valid.
Thus the induction hypothesis applies to
$(A^{(2)},U^{(2)},N_U^{(2)},B',h')$.

\paragraph{Residual active edges.}
We append the valid block produced by the second recursive call to the
one produced by the first.
By the definition of a valid block, their concatenation is again a
valid block starting from $A$.

It remains to bound the number of edges of the original set $U$ that
are still active.
Such an edge either belongs to $U\setminus U^{(1)}$, or belongs to
$U^{(2)}$ and survives the second recursive call.
Hence their number is at most
\[
\frac{N_U}{B^3}
+
\frac{N_U^{(2)}}{(B')^2}
=
\frac{N_U}{B^3}
+
\frac{N_U}{(B')^4}
<
\frac{N_U}{B^2},
\]
due to the choice of $B' = B^{1/2+1/\log^2\log B}$.
 
\paragraph{Number of layers and round complexity.} As in the proof of \cite[Lemma~5.4]{ghaffari_grunau_focs24}, a standard analysis of the recursion with two subproblems with parameter $B^{1/2+1/\log^2\log B}$ and $B\leq\log^{10}N$ being the base case shows that there are $O(\log B/\log\log N)$ base case calls. Each base case constructs $O(\log\log N)$ layers in $\widetilde O(\log N)$ rounds, so the total number of layers is $O(\log B)$, and the total round complexity is $\widetilde O(\log N\log B)$.
\end{proof}

\subsection{Completing the Deterministic Construction}

We now apply \Cref{lem:recursive-introvert} to the original graph and
complete the proof of \Cref{thm:introvert-decomposition-det}. 

\detintrovertdecomp*

\begin{proof} 
Let $G=(V,E)$ and let $H=\operatorname{Line}(G)$.
Choose a common polynomial upper bound $N$ such that
$|V(H)|=|E|<N$ and $N\geq2$.
Since we may take $N=\poly(n)$, we have $\log N=O(\log n)$.

For the initial recursive call, set
\[
A=U=E,\qquad N_U=N,\qquad B=N,
\]
and let $h(v)=0$ for every $v\in V(H)$.
The two invariants are immediate.
Indeed, $A\setminus U=\emptyset$, and, for every $e\in U$,
\(\bad_{h,d}^H(e)\leq |V(H)|<N=B\). 
Moreover, $|U|<N=N_U$ and
$\max_{v\in V(H)}h(v)=0\leq\Phi(B)$.
Thus, all assumptions of
\Cref{lem:recursive-introvert} are satisfied.

Applying \Cref{lem:recursive-introvert} gives an
$(\varepsilon,O(\log N))$-valid block of $O(\log N)$ layers starting
from $A=E$.
After this block, the number of edges of $U=E$ that remain active is
at most
\(\frac{N_U}{B^2}
=
\frac{1}{N}
<1\).

A valid block that
starts with active edge set $E$ and ends with no active edge forms a
layered introvert network decomposition: its layer edge sets partition
$E$, and every layer satisfies the required diameter and introvert
conditions.
We therefore obtain an
$(\varepsilon,O(\log N),O(\log N))$-layered introvert network
decomposition of $G$.

Finally, \Cref{lem:recursive-introvert} runs in
$\widetilde O(\log N\log B)=\widetilde O(\log^2 n)$ rounds.
Since $\log N=O(\log n)$, the resulting decomposition has parameters
$(\varepsilon,O(\log n),O(\log n))$, completing the proof.
\end{proof}

\section{Proof of the Small-Loss Clustering Lemma}
\label{app:small-loss-proof}

In this appendix, we prove the small-loss clustering lemma
(\Cref{lem:small-loss-base}) used in the base case of
\Cref{app:det-introvert}.
While both \Cref{app:det-introvert,app:small-loss-proof} are based on
adapting the algorithm of Ghaffari and
Grunau~\cite{ghaffari_grunau_focs24}, the nature of the adaptation is
different.

In \Cref{app:det-introvert}, the introvert condition forces us to take
into account all active edges, including those outside the set handled
by a recursive call.
This requires changing the invariants maintained by the recursion and
carefully checking that the modified invariants remain valid.
For this reason, we gave an essentially complete proof of the
recursive construction there.

The modification needed here is conceptually simpler.
The base case of Ghaffari and Grunau clusters at least half of the
relevant vertices, whereas we need the fraction left unclustered to be
an arbitrarily small constant $\rho>0$.
The particular constant $1/2$ is not essential to their construction.
Rather than reproducing their entire proof, we therefore
follow its structure and examine each step that contributes
to this constant loss, verifying that each such loss can be
made arbitrarily small.

\subsection{From Small Badness to Small Frontiers}
\label{app:small-loss-main}

We first restate the lemma to be proved.

\smalllossclustering*

We use the following small-loss version of
\cite[Theorem~A.1]{ghaffari_grunau_focs24}.
For a head start function $h$, the \emph{$r$-hop frontier} of a
vertex $u$ consists of the vertices whose shifted distance
$\dist_H(u,v)-h(v)$ is within $r$ of the minimum possible value.
Formally, it is the set of vertices $v$ satisfying
\[
\dist_H(u,v)-h(v)
\leq
\min_{w\in V(H)}
\left(
\dist_H(u,w)-h(w)
\right)+r.
\]

\begin{lemma}[Small-loss version of
{\cite[Theorem~A.1]{ghaffari_grunau_focs24}}]
\label{lem:gg-small-loss-a1}
Let $H$ be a graph, let $U\subseteq V(H)$, and let
$h:V(H)\to\mathbb N_{\geq0}$ be a head start function satisfying
\(\max_{v\in V(H)}h(v)=\widetilde O(\log N)\). 
Suppose that, for every $u\in U$, the
$\log^{10}\log N$-hop frontier of $u$ has size at most
$\log^{100}N$.
Then, for every constant $\rho>0$, there is a deterministic
$\widetilde O(\log N)$-round algorithm that computes a subset
$U'\subseteq U$ satisfying the following properties.
\begin{itemize}
    \item At most a $\rho$-fraction of $U$ is left unclustered:
    $|U\setminus U'|\leq\rho|U|$.
    \item Every connected component of $H[U']$ has strong diameter
    $O(\log N)$.
\end{itemize}
\end{lemma}

The proof of \Cref{lem:gg-small-loss-a1} is deferred to \Cref{app:small-loss-a1}.
We first show that it implies \Cref{lem:small-loss-base}.

\begin{proof}[Proof of \Cref{lem:small-loss-base}]
Let
\[
q=\left\lceil\log^{20}\log N\right\rceil
\]
and define the scaled head start function
$\widetilde h(v)=q h(v)$.
We verify that $\widetilde h$ satisfies the assumptions of
\Cref{lem:gg-small-loss-a1}.

Fix $u\in U$ and consider a vertex $v$ in the
$\log^{10}\log N$-hop frontier of $u$ with respect to
$\widetilde h$.
First, $v$ is at distance at most $d$ from $u$.
Indeed, by taking $w=u$ in the definition of the frontier,
\[
\dist_H(u,v)-\widetilde h(v)
\leq
-\widetilde h(u)+\log^{10}\log N,
\]
and hence
\[
\dist_H(u,v)
\leq
\widetilde h(v)-\widetilde h(u)+\log^{10}\log N
\leq
q\log N+\log^{10}\log N
\leq d,
\]
where the last inequality holds for a sufficiently large choice of
the constant $c$ in the definition of $d$.

Next, among all vertices at distance
$r=\dist_H(u,v)$ from $u$, the vertex $v$ must have maximum head
start with respect to $h$.
Indeed, if some vertex $w$ at the same distance satisfied
$h(w)>h(v)$, then
\[
\widetilde h(w)-\widetilde h(v)
\geq q
>
\log^{10}\log N,
\]
which contradicts the assumption that $v$ belongs to the frontier.

Consequently, for every distance $r\leq d$, the number of frontier
vertices at distance $r$ from $u$ is at most
$\bad_{h,d}^H(u)\leq\log^{10}N$.
Therefore, the entire frontier has size at most
\[
(d+1)\log^{10}N
\leq
\log^{100}N.
\]
Moreover,
\[
\max_{v\in V(H)}\widetilde h(v)
\leq
q\log N
=
\widetilde O(\log N).
\]

Thus $\widetilde h$ satisfies all assumptions of
\Cref{lem:gg-small-loss-a1}.
Applying that lemma gives a subset $U'\subseteq U$ satisfying the required two properties in $\widetilde O(\log N)$ rounds.
\end{proof}

\subsection{Reducing the Frontier Size}
\label{app:small-loss-a1}

The proof of \cite[Theorem~A.1]{ghaffari_grunau_focs24} consists of
two steps, captured by their Theorems~A.2 and A.3.
We use the following small-loss versions of these two results.

\begin{lemma}[Small-loss version of
{\cite[Theorem~A.2]{ghaffari_grunau_focs24}}]
\label{lem:gg-small-loss-a2}
There is an absolute constant $c$ for which the following holds.
Let $H$ be a graph, let $U\subseteq V(H)$, and let
$h:V(H)\to\mathbb N_{\geq0}$ be a head start function satisfying
$\max_{v\in V(H)}h(v)=\widetilde O(\log N)$.
Suppose that, for every $u\in U$, the
$\log^{10}\log N$-hop frontier of $u$ has size at most
$\log^{100}N$.
Then, for every constant $\rho>0$, there is a deterministic
$\widetilde O(\log N)$-round algorithm that computes a subset
$U'\subseteq U$ and a head start function
$h':V(H)\to\mathbb N_{\geq0}$ satisfying the following properties.
\begin{itemize}
    \item $|U\setminus U'|\leq\rho|U|$.
    \item $\max_{v\in V(H)}h'(v)=\widetilde O(\log N)$.
    \item For all $u\in U'$, the $\log^2\log N$-hop frontier of $u$
    with respect to $h'$ has size at most $(\log\log N)^c$.
\end{itemize}
\end{lemma}

\begin{proof}
We use exactly the algorithm and analysis of
\cite[Theorem~A.2]{ghaffari_grunau_focs24}.
The only difference between \Cref{lem:gg-small-loss-a2} and \cite[Theorem~A.2]{ghaffari_grunau_focs24} is that Ghaffari and Grunau fix $\rho = 0.1$.
In their notation, \cite[Claim~A.7]{ghaffari_grunau_focs24} shows that, after iteration $i$, the set
$U_i$ satisfies
\(|U_i|
\geq
\left(1-\frac{2i}{\log^{10}\log N}\right)|U|\).
The algorithm performs $O(\log^2\log N)$ iterations.
Hence the fraction of vertices discarded is already
$O(1/\log^8\log N)$, which is not just at most 0.1 but also at most any constant $\rho$, as required.
\end{proof}

\begin{lemma}[Small-loss version of
{\cite[Theorem~A.3]{ghaffari_grunau_focs24}}]
\label{lem:gg-small-loss-a3}
For every constant $c\geq1$, the following holds.
Let $H$ be a graph, let $U\subseteq V(H)$, and let
$h:V(H)\to\mathbb N_{\geq0}$ be a head start function satisfying
$\max_{v\in V(H)}h(v)=\widetilde O(\log N)$.
Suppose that, for every $u\in U$, the $\log^2\log N$-hop frontier of
$u$ has size at most $(\log\log N)^c$.
Then, for every constant $\rho>0$, there is a deterministic
$\widetilde O(\log N)$-round algorithm that computes a subset
$U'\subseteq U$ satisfying the following properties.
\begin{itemize}
    \item At most a $\rho$-fraction of $U$ is left unclustered:
    $|U\setminus U'|\leq\rho|U|$.
    \item Every connected component of $H[U']$ has strong diameter
    $O(\log N)$.
\end{itemize}
\end{lemma}

We prove \Cref{lem:gg-small-loss-a3} in \Cref{app:small-loss-a3}.
We can now complete the proof of the small-loss version of
\cite[Theorem~A.1]{ghaffari_grunau_focs24}.

\begin{proof}[Proof of \Cref{lem:gg-small-loss-a1}]
Apply \Cref{lem:gg-small-loss-a2} with parameter $\rho/2$.
This gives a subset $U_1\subseteq U$ and a head start function $h'$
such that $|U\setminus U_1|\leq(\rho/2)|U|$, and the
$\log^2\log N$-hop frontier of every $u\in U_1$ with respect to $h'$
has size at most $(\log\log N)^c$.

We can therefore apply \Cref{lem:gg-small-loss-a3} to $U_1$ and $h'$,
again with parameter $\rho/2$.
Let $U'\subseteq U_1$ be the resulting set.
Then
\(
|U'|
\geq
\left(1-\frac{\rho}{2}\right)|U_1|
\geq
\left(1-\frac{\rho}{2}\right)^2|U|
\geq
(1-\rho)|U|\).
Moreover, every connected component of $H[U']$ has strong diameter
$O(\log N)$.
Both applications take $\widetilde O(\log N)$ rounds, so the total
round complexity is also $\widetilde O(\log N)$.
\end{proof}

\subsection{From Small Frontiers to Low-Diameter Clusters}
\label{app:small-loss-a3}

The proof of \cite[Theorem~A.3]{ghaffari_grunau_focs24} is a direct
application of their Theorem~A.11.
Before stating the small-loss version that we need, we introduce two
notions used in that theorem.

\paragraph{Clustering terminology.}
For a vertex $u$ and a cluster $C$ of $H$, let
\[
\dist_H(u,C)=\min_{v\in C}\dist_H(u,v).
\]
For two clusters $C$ and $C'$, let
\[
\dist_H(C,C')
=
\min_{u\in C,\;v\in C'}\dist_H(u,v).
\]
A clustering $\mathcal C$ is \emph{$s$-separated} if
$\dist_H(C,C')\geq s$ for every two distinct clusters
$C,C'\in\mathcal C$.

For a vertex $u$, the \emph{$s$-hop degree} of $u$ with respect to
$\mathcal C$ is the number of clusters $C\in\mathcal C$ satisfying
$\dist_H(u,C)\leq s$.
The \emph{$s$-hop degree} of $\mathcal C$ is the maximum $s$-hop
degree over all vertices clustered by $\mathcal C$.
These definitions agree with
\cite[Definition~A.2]{ghaffari_grunau_focs24}.

We only need Theorem~A.11 in the particular parameter regime arising
in the proof of Theorem~A.3, so we state its small-loss version directly
in this regime.

\begin{lemma}[Small-loss version of
{\cite[Theorem~A.11]{ghaffari_grunau_focs24}}]
\label{lem:gg-small-loss-a11}
For every constant $c\geq1$, the following holds. Set
\[
s=\left\lfloor\frac{\log^2\log N}{3}\right\rfloor
\qquad\text{and}\qquad
\mathrm{DEG}=(\log\log N)^c.
\]
Let $\mathcal C$ be a clustering of a graph $H$ with weak diameter
$\widetilde O(\log N)$ and $s$-hop degree at most $\mathrm{DEG}$.
Let $R$ be the set of vertices clustered by $\mathcal C$.
For every constant $\rho>0$, there is a deterministic
$\widetilde O(\log N)$-round algorithm that computes a clustering
$\mathcal C_{\mathrm{out}}$ satisfying the following properties.
\begin{itemize}
    \item Every cluster has strong diameter $O(\log N)$.
    \item The clustering is $2$-separated.
    \item All vertices clustered by $\mathcal C_{\mathrm{out}}$
    belong to $R$, and at least $(1-\rho)|R|$ vertices of $R$ are
    clustered by $\mathcal C_{\mathrm{out}}$.
\end{itemize}
\end{lemma}

The proof of \Cref{lem:gg-small-loss-a11} is given in
\Cref{app:small-loss-a11}.
We first show that it implies \Cref{lem:gg-small-loss-a3}.

\begin{proof}[Proof of \Cref{lem:gg-small-loss-a3}]
We follow the proof of
\cite[Theorem~A.3]{ghaffari_grunau_focs24}.
Let $\mathcal C_h$ be the head start clustering obtained by assigning
every vertex $u\in V(H)$ to a vertex $v\in V(H)$ minimizing
$\dist_H(u,v)-h(v)$, with ties broken consistently.
Since $\max_{v\in V(H)}h(v)=\widetilde O(\log N)$,
$\mathcal C_h$ has strong diameter $\widetilde O(\log N)$.

Set $s$ and $\mathrm{DEG}$ as in the statement of
\Cref{lem:gg-small-loss-a11}.
We claim that every vertex $u\in U$ has $s$-hop degree at most
$\mathrm{DEG}$ with respect to $\mathcal C_h$.

To see this, consider a cluster of $\mathcal C_h$ with center $v$
whose distance from $u$ is at most $s$.
Let $x$ be a vertex of this cluster with $\dist_H(u,x)\leq s$, and
let $w$ minimize $\dist_H(u,w)-h(w)$.
Then
\begin{align*}
\dist_H(u,v)-h(v)
&\leq
\dist_H(u,x)+\dist_H(x,v)-h(v)
&& \text{by the triangle inequality} \\
&\leq
\dist_H(u,x)+\dist_H(x,w)-h(w)
&& \text{since $x$ is assigned to $v$} \\
&\leq
\dist_H(u,x)+\dist_H(u,w)+\dist_H(u,x)-h(w)
&& \text{by the triangle inequality} \\
&\leq
2s+\dist_H(u,w)-h(w)
&& \text{since $\dist_H(u,x)\leq s$}.
\end{align*}
Thus the center $v$ belongs to the $2s$-hop frontier of $u$.
Since $2s\leq\log^2\log N$, every cluster within distance $s$ of $u$
has a distinct center in the $\log^2\log N$-hop frontier of $u$.
By the assumption of \Cref{lem:gg-small-loss-a3}, this frontier has
size at most $(\log\log N)^c=\mathrm{DEG}$.
Hence the $s$-hop degree of $u$ is at most $\mathrm{DEG}$.

Let $\mathcal C$ be obtained from $\mathcal C_h$ by retaining only
the vertices in $U$ in each cluster.
Then $\mathcal C$ clusters exactly $U$, has weak diameter
$\widetilde O(\log N)$, and has $s$-hop degree at most
$\mathrm{DEG}$.
We may therefore apply \Cref{lem:gg-small-loss-a11} to $\mathcal C$
with parameter $\rho$.

Let $U' \subseteq U$ be the set of vertices clustered by the resulting clustering
$\mathcal C_{\mathrm{out}}$.
By  \Cref{lem:gg-small-loss-a11}, 
$|U\setminus U'|\leq\rho|U|$, and $\mathcal C_{\mathrm{out}}$ is $2$-separated, so the
connected components of $H[U']$ are precisely its clusters and hence
have strong diameter $O(\log N)$.
The round complexity is $\widetilde O(\log N)$, as required.
\end{proof}

\subsection{Iterative Small-Boundary Clustering}
\label{app:small-loss-a11}

We now prove \Cref{lem:gg-small-loss-a11}.
The proof follows
\cite[Theorem~A.11]{ghaffari_grunau_focs24}.
The main modification is in
\cite[Lemma~A.14]{ghaffari_grunau_focs24}, which extracts a reasonably
large clustering whose boundary is at most a fixed constant fraction
of its size.
We show that this constant can be made arbitrarily small.
The preceding subsampling step,
\cite[Lemma~A.12]{ghaffari_grunau_focs24}, is used without any
modification. 

\begin{lemma}[Small-loss version of
{\cite[Lemma~A.14]{ghaffari_grunau_focs24}}]
\label{lem:gg-small-loss-a14}
For all constants $c\geq1$ and $\delta>0$, the following holds.
Set
\[
s=\left\lfloor\frac{\log^2\log N}{3}\right\rfloor
\qquad\text{and}\qquad
\mathrm{DEG}=(\log\log N)^c.
\]
Let $\mathcal C$ be a clustering of a graph $H$ with weak diameter
$\widetilde O(\log N)$ and $s$-hop degree at most $\mathrm{DEG}$,
and let $R$ be the set of vertices clustered by $\mathcal C$.
There is a deterministic $\widetilde O(\log N)$-round algorithm that
computes a $2$-separated clustering $\mathcal D$ using only vertices
of $R$.
Let $X\subseteq R$ be the set of vertices clustered by $\mathcal D$.
Then the following properties hold.
\begin{itemize}
    \item $|X|\geq |R|/(16\mathrm{DEG})$.
    \item Every cluster of $\mathcal D$ has weak diameter
    $\widetilde O(\log N)$.
    \item At most $\delta|X|$ vertices of $R\setminus X$ have a
    neighbor in $X$.
\end{itemize}
\end{lemma}

\begin{proof}
We follow the proof of
\cite[Lemma~A.14]{ghaffari_grunau_focs24}.
First apply \cite[Lemma~A.12]{ghaffari_grunau_focs24}:
using $\widetilde O(s\log^2(\mathrm{DEG})\log N)$ rounds, it computes an
$s$-separated clustering $\mathcal C'$ of weak diameter
$\widetilde O(s\log N)$ that clusters at least
$|R|/(8\mathrm{DEG})$ vertices, all belonging to $R$.
For each $C'\in\mathcal C'$ and $k\geq0$, define
\[
C'_{\leq k}
=
\left\{
v\in R:\dist_H(v,C')\leq k
\right\}.
\]
For each $C'$, we look for a radius
$k\leq\lfloor s/3\rfloor$ satisfying
\[
|C'_{\leq k+1}|
\leq
(1+\delta)|C'_{\leq k}|.
\]
If no such radius exists, then
\[
|C'_{\leq\lfloor s/3\rfloor}|
>
(1+\delta)^{\lfloor s/3\rfloor}|C'|.
\]
Since $\mathcal C'$ is $s$-separated, the sets
$C'_{\leq\lfloor s/3\rfloor}$ are pairwise disjoint.
Hence the total number of vertices belonging to clusters $C'$ for
which no suitable radius exists is at most
\[
(1+\delta)^{-\lfloor s/3\rfloor}|R|
\leq
\frac{|R|}{16\mathrm{DEG}},
\]
where the inequality follows from
$s=\Theta(\log^2\log N)$ and
$\mathrm{DEG}=(\log\log N)^c$.

Since $\mathcal C'$ initially clusters at least
$|R|/(8\mathrm{DEG})$ vertices, after discarding the clusters for
which no suitable radius exists, at least
$|R|/(16\mathrm{DEG})$ vertices remain.
For every remaining cluster $C'$, choose such a radius $k$ and output
$C'_{\leq k}$.
These output clusters remain $2$-separated and have weak diameter
$\widetilde O(\log N)$.

Finally, every vertex of $R$ outside the output clustering that is
adjacent to $C'_{\leq k}$ belongs to
$C'_{\leq k+1}\setminus C'_{\leq k}$.
By the choice of $k$, this set has size at most
$\delta|C'_{\leq k}|$.
Summing over all output clusters shows that at most $\delta|X|$
vertices of $R\setminus X$ are adjacent to $X$.

The only modification from the proof of
\cite[Lemma~A.14]{ghaffari_grunau_focs24} is replacing its fixed
growth factor by $1+\delta$; all other steps are unchanged.
The call to Lemma~A.12 takes
$\widetilde O(s\log^2(\mathrm{DEG})\log N)$ rounds, while the
subsequent ball growing and carving take
$\widetilde O(s\log N) + O(s)$ rounds.
By the choice of  $s$ and $\mathrm{DEG}$, the total round complexity
is $\widetilde O(\log N)$.
\end{proof}

\begin{proof}[Proof of \Cref{lem:gg-small-loss-a11}]
It suffices to consider $0<\rho<1$.
Set $\delta=\rho/4$.
We follow the iterative construction in the proof of
\cite[Theorem~A.11]{ghaffari_grunau_focs24}.

Let $\mathcal C_0=\mathcal C$ and let
$\mathcal C_0^{\mathrm{out}}$ be empty.
In iteration $i$, apply \Cref{lem:gg-small-loss-a14} to
$\mathcal C_i$, obtaining a $2$-separated clustering
$\mathcal D_i$.
Add the clusters of $\mathcal D_i$ to
$\mathcal C_i^{\mathrm{out}}$ to obtain
$\mathcal C_{i+1}^{\mathrm{out}}$.
Then obtain $\mathcal C_{i+1}$ from $\mathcal C_i$ by removing all
vertices clustered by $\mathcal D_i$, together with all remaining
vertices that have a neighbor in one of these clusters.

For a clustering $\mathcal A$, write $V(\mathcal A)$ for the set of
vertices clustered by $\mathcal A$.
By \Cref{lem:gg-small-loss-a14},
\[
|V(\mathcal C_{i+1}^{\mathrm{out}})|
-
|V(\mathcal C_i^{\mathrm{out}})|
\geq
\frac{|V(\mathcal C_i)|}{16\mathrm{DEG}}.
\]
All newly clustered vertices are removed from $\mathcal C_i$.
Therefore,
\[
|V(\mathcal C_i)|-|V(\mathcal C_{i+1})|
\geq
|V(\mathcal C_{i+1}^{\mathrm{out}})|
-
|V(\mathcal C_i^{\mathrm{out}})|,
\]
and hence
\[
|V(\mathcal C_{i+1})|
\leq
\left(
1-\frac{1}{16\mathrm{DEG}}
\right)
|V(\mathcal C_i)|.
\]

Deleting vertices from the input clustering cannot increase its weak
diameter or its $s$-hop degree, so
\Cref{lem:gg-small-loss-a14} remains applicable in every iteration.
Also, before proceeding to the next iteration we remove all vertices
adjacent to the newly extracted clusters.
Thus clusters extracted in different iterations are non-adjacent,
and $\mathcal C_i^{\mathrm{out}}$ remains $2$-separated.

Run the procedure for
\[
T=
\left\lceil
16\mathrm{DEG}\ln\frac{4}{\rho}
\right\rceil
\]
iterations.
Then
\[
|V(\mathcal C_T)|
\leq
\left(
1-\frac{1}{16\mathrm{DEG}}
\right)^T
|R|
\leq
\frac{\rho}{4}|R|.
\]
Across all iterations, the total number of vertices discarded because
they are adjacent to extracted clusters is at most
$\delta|V(\mathcal C_T^{\mathrm{out}})|$.
Consequently,
\[
|R|
\leq
(1+\delta)|V(\mathcal C_T^{\mathrm{out}})|
+
|V(\mathcal C_T)|,
\]
and therefore, using $\delta=\rho/4$,
\[
|V(\mathcal C_T^{\mathrm{out}})|
\geq
\frac{1-\rho/4}{1+\rho/4}|R|
\geq
\left(1-\frac{\rho}{2}\right)|R|.
\]
Thus $\mathcal C_T^{\mathrm{out}}$ is a $2$-separated clustering of
weak diameter $\widetilde O(\log N)$ that clusters at least a
$(1-\rho/2)$-fraction of $R$.

It remains only to obtain strong diameter.
We use the well-known fact that, for every
graph on at most $N$ vertices and every $\beta>0$, one can remove at
most a $\beta$-fraction of the vertices so that every remaining
connected component has strong diameter
$O\left(\frac{\log N}{\beta}\right)$; see, e.g.,
\cite{ChangLi,ElkinN16}.
Apply this fact to the subgraph induced by each cluster of
$\mathcal C_T^{\mathrm{out}}$, with $\beta=\rho/2$.
Since these clusters have weak diameter $\widetilde O(\log N)$, each
induced subgraph can be gathered and the decomposition computed
locally in $\widetilde O(\log N)$ rounds, with all clusters processed
in parallel.
The resulting clustering remains $2$-separated, has strong diameter
$O(\log N)$, and clusters at least
$\left(1-\frac{\rho}{2}\right)
|V(\mathcal C_T^{\mathrm{out}})|
\geq
\left(1-\frac{\rho}{2}\right)^2|R|
\geq
(1-\rho)|R|$
vertices.

Since $T=O(\mathrm{DEG})$ and
$\mathrm{DEG}=(\log\log N)^c$, the total round complexity remains
$\widetilde O(\log N)$.
\end{proof}

\end{document}